\documentclass[twoside,11pt]{article}

\usepackage[utf8]{inputenc} 
\usepackage[T1]{fontenc}    

\usepackage{amsmath}
\usepackage{amssymb}
\usepackage{amsthm}
\usepackage{thm-restate}
\usepackage{subcaption}

\usepackage[preprint]{jmlr2e}

\usepackage{url}            
\usepackage{booktabs}       
\usepackage{amsfonts}       
\usepackage{nicefrac}       
\usepackage{microtype}      
\usepackage{xcolor}         

\usepackage[ruled]{algorithm2e}
\usepackage{multirow}
\usepackage{tikz}
\usetikzlibrary{shapes,arrows.meta,positioning,backgrounds,fit}

\graphicspath{{figures/}}

\declaretheorem[name=Theorem]{theorem}
\declaretheorem[name=Lemma, sibling=theorem]{lemma}
\declaretheorem[name=Definition, sibling=theorem]{definition}

\declaretheorem[name=Condition, sibling=theorem]{condition}

\newcommand{\calx}{\mathcal{X}}

\newcommand{\caln}{\mathcal{N}}
\newcommand{\calm}{\mathcal{M}}
\newcommand{\calq}{\mathcal{Q}}

\newcommand{\R}{\mathbb{R}}
\newcommand{\E}{\mathbb{E}}
\newcommand{\I}{\mathbb{I}}
\newcommand{\Var}{\mathrm{Var}}
\newcommand{\dx}{\,\mathrm{d}}
\newcommand{\sdp}{\tilde{s}}
\newcommand{\xsyn}{X^{Syn}}
\newcommand{\xreal}{X}
\newcommand{\xpred}{X^*}
\newcommand{\xpredbar}{\bar{X}^*}
\newcommand{\obs}{Z}
\newcommand{\sigmadp}{\sigma_{DP}}
\newcommand{\med}{\mathrm{MED}}
\newcommand{\medt}{\med_\theta}
\newcommand{\KL}[2]{\mathrm{KL}(#1 \, || \, #2)}

\newcommand{\convprob}{\xrightarrow{P}}

\DeclareMathOperator{\TV}{TV}
\newcommand{\invchisq}{\text{Inv-}\chi^2}
\newcommand{\ups}[1]{\bar{#1}}
\newcommand{\ds}[1]{\hat{#1}}
\newcommand{\toydatadelta}{= 2.5\cdot 10^{-7}}
\newcommand{\adultdelta}{\approx 4.7\cdot 10^{-10}}
\newcommand{\qex}{{\bar{Q}^+}}
\newcommand{\qap}{{\bar{Q}}}
\newcommand{\pap}{{\bar{p}}}

\usepackage{lastpage}
\jmlrheading{23}{2023}{1-\pageref{LastPage}}{10/23}{}{21-0000}{Ossi Räisä, Joonas Jälkö and Antti Honkela}

\ShortHeadings{On Consistent Bayesian Inference from Synthetic Data}{Räisä, Jälkö and Honkela}
\firstpageno{1}

\begin{document}

\title{On Consistent Bayesian Inference from Synthetic Data}

\author{\name Ossi Räisä \email ossi.raisa@helsinki.fi\\
\name Joonas Jälkö \email joonas.jalko@helsinki.fi\\
\name Antti Honkela \email antti.honkela@helsinki.fi\\
\addr Department of Computer Science\\
University of Helsinki\\
P.O. Box 68 (Pietari Kalmin katu 5) 
\\ 00014 University of Helsinki, Finland
}

\editor{TBD}

\maketitle

\begin{abstract}
    Generating synthetic data, with or without differential privacy, has 
    attracted significant attention as a potential solution to 
    the dilemma between making data easily available, and the privacy of data 
    subjects. Several works have shown that consistency of 
    downstream analyses from
    synthetic data, including accurate uncertainty estimation, requires accounting 
    for the synthetic data generation.
    There are very few methods of doing so, most of them for frequentist
    analysis. In this paper, we study how to perform consistent Bayesian
    inference from synthetic data.  We prove that mixing 
    posterior samples obtained separately from multiple large synthetic 
    data sets converges to the posterior of the downstream analysis
    under standard regularity 
    conditions when the analyst's model is compatible with the data provider's 
    model. We also present several examples showing how the theory works in practice,
    and showing how Bayesian inference can fail when the compatibility assumption is 
    not met, or the synthetic data set is not significantly larger than the 
    original.
\end{abstract}

\begin{keywords}
  synthetic data, Bayesian inference, Bernstein-von Mises theorem, differential privacy
\end{keywords}

\section{Introduction}
Synthetic data has the potential of opening privacy-sensitive data sets for 
widespread analysis.
The idea is to train a generative model with real data, and release 
synthetic data that has been generated from the model. The synthetic data 
does not contain records from real people, and ideally it
preserves the population-level properties of the real data, 
making it useful for analysis. Privacy preservation can be guaranteed with 
\emph{differential privacy} (DP) \citep{dworkCalibratingNoiseSensitivity2006},
which offers provable protection of privacy.

The most convenient and straightforward
way for downstream analysts to analyse synthetic data is using the same method 
that would be used with real data. However, ignoring the additional 
stochasticity arising 
from the synthetic data generation will 
yield biased results and overconfident uncertainty 
estimates~\citep{raghunathanMultipleImputationStatistical2003,
wildeFoundationsBayesianLearning2021,
raisaNoiseAwareStatisticalInference2023}. 
This is especially problematic under DP, which requires adding extra noise,
which will be ignored if the synthetic data is treated like real data.
This problem creates the need for \emph{noise-aware} analyses that account 
for the synthetic data generation.

For frequentist downstream analyses, it is possible to account for the 
synthetic data generation by generating and analysing multiple synthetic
data sets~\citep{raghunathanMultipleImputationStatistical2003}. Recent work 
has extended this to DP synthetic 
data~\citep{raisaNoiseAwareStatisticalInference2023},
which allows generating multiple synthetic data sets without compromising on privacy. 
These methods reuse the analysis method for the real data, and only require 
using simple combining rules to combine the results from the analyses on each 
synthetic data set, making them simple to apply.

For Bayesian 
downstream analyses, \citet{wildeFoundationsBayesianLearning2021} have shown 
that the analyst can use additional samples of public real data to correct their 
analysis. However, their method requires targeting a generalised notion of 
the posterior~\citep{bissiriGeneralFrameworkUpdating2016}
and needs the additional public data for calibration.
\citet{ghalebikesabiMitigatingStatisticalBias2022} propose a correction using 
importance sampling to avoid the need of public data, but only prove convergence 
to a generalised posterior and do not clearly address the noise-awareness of the 
method.

The simple frequentist methods using multiple synthetic data sets were derived from 
methods in missing data imputation~\citep{rubinMultipleImputationNonresponse1987}, so
our starting point is the method of \citet{gelmanBayesianDataAnalysis2004,gelmanBayesianDataAnalysis2014}
for Bayesian inference with missing data. They
proposed inferring the downstream posterior by imputing multiple completed 
data sets, inferring the analysis posterior for each completed data set separately,
and mixing the posteriors together. We investigate whether this method is also 
applicable to synthetic data, generated with or without DP.

\subsection{Contributions}
\begin{enumerate}
    \item We study inferring the downstream analysis posterior by generating
    multiple synthetic data sets, inferring the analysis posterior for 
    each synthetic data set as if it were the real data set, and mixing the 
    posteriors together. We find two important conditions for consistent 
    Bayesian inference with this method: synthetic data sets that are larger 
    than the original one, and a notion of compatibility between the 
    data provider's and analyst's models called 
    \emph{congeniality}~\citep{mengMultipleImputationInferencesUncongenial1994}.
    \item We prove that when congeniality is met and the Bernstein--von Mises
    theorem applies, this method converges to the
    true posterior as the number of synthetic data sets and the
    size of the synthetic data sets grow. Under stronger assumptions, we prove 
    a convergence rate for this method in the synthetic 
    data set size, which we expect to match the rate that 
    usually applies in the Bernstein--von Mises 
    theorem~\citep{hippBernsteinvMisesApproximation1976}. 
    These are presented in 
    Section~\ref{sec:syn-data-bayes-inference}.
    \item We evaluate this method with two examples in 
    Sections~\ref{sec:gauss-example} and \ref{sec:logistic-regression-example}: 
    non-private univariate Gaussian mean or variance estimation, and DP
    Bayesian logistic regression. In the first example, we use the tractability 
    of the model to derive further theoretical properties of the method,
    and in both examples, we verify that the method works in practice 
    when the assumptions are met, and examine what can happen when
    they are not met. Our code is available under an open-source 
    license\footnote{\url{https://github.com/DPBayes/NAPSU-MQ-bayesian-downstream-experiments}}.
\end{enumerate}

\subsection{Related Work}

Generating synthetic data to preserve privacy was, as far as we know, 
originally proposed by \citet{liewDataDistortionProbability1985}.
\citet{rubin1993statistical} proposed accounting for the 
synthetic data generation in frequentist downstream analyses by adapting 
\emph{multiple imputation}~\citep{rubinMultipleImputationNonresponse1987},
which involves generating multiple synthetic data sets, analysing each of them,
and combining the results with so called Rubin's 
rules~\citep{raghunathanMultipleImputationStatistical2003,reiter2002satisfying}.
Recently, \citet{raisaNoiseAwareStatisticalInference2023} have shown that 
multiple imputation also works for synthetic data generated under DP 
when the data generation algorithm is \emph{noise-aware} in a certain sense.

Recently, \citet{vanbreugelSyntheticDataReal2023} have studied uncertainty 
quantification for prediction tasks when training on synthetic data. They 
propose generating multiple synthetic data sets, like the multiple imputation 
line of work, and experimentally show that aggregating predictions from the 
multiple synthetic data sets improves generalisation performance and uncertainty 
quantification. They use a similar Bayesian framework as we do to justify 
using multiple synthetic data sets, but their theoretical study of the framework 
is very light. In particular, they do not consider the effect of the synthetic 
data set size theoretically.

\citet{wildeFoundationsBayesianLearning2021} study downstream Bayesian 
inference from DP synthetic data by considering the analyst's model to be 
misspecified, and targeting a generalised notion of the 
posterior~\citep{bissiriGeneralFrameworkUpdating2016} to deal with the 
misspecification, which makes their method more difficult to apply than standard 
Bayesian inference. They also assume that the analyst has additional public data 
available to calibrate their method. 

\citet{ghalebikesabiMitigatingStatisticalBias2022} use importance sampling 
to correct for bias with DP synthetic data, and have Bayesian inference as 
an example application. 
However, they also target a generalised 
variant~\citep{bissiriGeneralFrameworkUpdating2016} of the posterior instead 
of the noise-aware posterior we target, and they do not evaluate uncertainty 
estimation, so the noise-awareness of their method is not clear.

We are not aware of any existing work adapting multiple imputation 
for Bayesian downstream analysis in the synthetic data setting. In the missing 
data setting without DP, where multiple imputation was originally 
developed~\citep{rubinMultipleImputationNonresponse1987}, 
\citet{gelmanBayesianDataAnalysis2004,gelmanBayesianDataAnalysis2014} have proposed sampling the downstream 
posterior by mixing samples of the downstream posteriors from each of the 
multiple synthetic data sets. We find that this is not sufficient in the 
synthetic data setting, and add one extra component: 
our synthetic data sets are larger than the original data set.
We compare the two cases in 
more detail in Section~\ref{sec:relation-with-imputation},
and in particular explain why large synthetic data sets are not needed 
in the missing data setting.

Noise-aware DP Bayesian inference is critical for taking into account the DP noise in
synthetic data, but only a few works address this even without synthetic data.
\citet{bernsteinDifferentiallyPrivateBayesian2018} present an inference method for simple 
exponential family models. Their approach was extended to linear 
models~\citep{bernsteinDifferentiallyPrivateBayesian2019} and 
generalised linear models~\citep{kulkarniDifferentiallyPrivateBayesian2021}.
Recently, \citet{juDataAugmentationMCMC2022} developed an MCMC sampler that 
can sample the noise-aware posterior using a noisy summary statistic.

\section{Background}\label{sec:background}
In this section, we introduce some background needed for the rest of our work.
We start by introducing Bayesian inference and the Bernstein--von Mises theorem 
in Section~\ref{sec:background-bayes-inference}, and then introduce differential 
privacy in Section~\ref{sec:differential-privacy} and noise-aware synthetic data generation in 
Section~\ref{sec:background-noise-aware-syn-data}.

\subsection{Bayesian Inference}\label{sec:background-bayes-inference}
Bayesian inference is a paradigm of statistical inference where the data 
analyst's uncertainty in a quantity $Q$ after observing data $X$ 
is represented using the posterior distribution 
$p(Q | X)$~\citep{gelmanBayesianDataAnalysis2014}. The posterior 
is given by Bayes' rule:
\begin{equation*}
    p(Q | X) = \frac{p(X | Q)p(Q)}{\int p(X | Q')p(Q')\dx Q'},
\end{equation*}
where $p(X | Q)$ is the likelihood of observing the data $X$ for a given value of 
$Q$, and $p(Q)$ is the analyst's prior of $Q$. Computing the denominator is 
typically intractable, so analysts often use numerical methods to sample 
$p(Q | X)$~\citep{gelmanBayesianDataAnalysis2014}.

It turns out that in many typical 
settings, the prior's influence on the posterior vanishes when the data set 
$X$ is large. A basic example of this is the Bernstein--von Mises 
theorem~\citep{vandervaartAsymptoticStatistics1998},
which informally states that under some regularity conditions, the posterior 
approaches a Gaussian that does not depend on the prior as the size of the 
data set increases.

A crucial component of the theorem, and also our theory, is the notion of
\emph{total variation distance} which is used to measure the difference between 
two random variables or probability distributions.
\begin{restatable}{definition}{definitiontotalvariationdistance}\label{def:total-variation-distance}
    The total variation distance between random variables (or distributions) 
    $P_1$ and $P_2$ is
    \begin{equation*}
        \TV(P_1, P_2) = \sup_A |\Pr(P_1 \in A) - \Pr(P_2 \in A)|,
    \end{equation*}
    where $A$ is any measurable set.
\end{restatable}
As a slight abuse of notation, we allow the arguments of $\TV(\cdot, \cdot)$ 
to be random variables, probability distributions, or probability 
density functions interchangeably.
We list some properties of total variation distance that we use in 
Lemma~\ref{thm:total-variation-distance-properties} in 
Appendix~\ref{sec:total-variation-distance-properties}.

Now we can state the theorem. 
\begin{restatable}[Bernstein--von Mises, \citealp{vandervaartAsymptoticStatistics1998}]{theorem}{theorembvm}\label{thm:bvm}
    Let $n$ denote the size of the data set $X_n$. 
    Under regularity conditions stated in Condition~\ref{cond:bvm} in 
    Appendix ~\ref{sec:bvm-regularity-conditions}, for true parameter value 
    $Q_0$, the posterior $\bar{Q}(X_n) \sim p(Q | X_n)$ satisfies 
    \begin{equation*}
        \TV\left(\sqrt{n}(\bar{Q}(X_n) - Q_0), \caln(\mu(X_n), \Sigma)\right)
        \convprob 0
    \end{equation*}
    as $n\to \infty$ for some $\mu(X_n)$ and $\Sigma$, that do not depend on the 
    prior, where the convergence in probability is over sampling 
    $X_n \sim p(X_n | Q_0)$.
\end{restatable}

\subsection{Differential Privacy and Noise-Aware Synthetic Data}\label{sec:differential-privacy}
\emph{Differential privacy} (DP)~\citep{dworkCalibratingNoiseSensitivity2006} 
quantifies the privacy loss from releasing the results of 
analysing data. The quantification is done by looking at the output distributions
of the analysis algorithm for two data sets that differ in a single 
data subject~\citep{dworkAlgorithmicFoundationsDifferential2014}:
\begin{definition}
    An algorithm $\calm$ is $(\epsilon, \delta)$-DP if 
    \begin{equation*}
        \Pr(\calm(X) \in S) \leq e^\epsilon \Pr(\calm(X') \in S) + \delta
    \end{equation*}
    for all measurable sets $S$ and all data sets $X, X'$ that differ in one 
    data subject.
\end{definition}
The choice of $\epsilon$ and $\delta$ is a matter of 
policy~\citep{dworkDifferentialPrivacySurvey2008}. 
One should set $\delta \ll \nicefrac{1}{n}$ for $n$ datapoints, as
$\delta \approx \nicefrac{1}{n}$ permits mechanisms that clearly 
violate privacy~\citep{dworkAlgorithmicFoundationsDifferential2014}.

A common primitive for making an algorithm DP is the 
\emph{Gaussian mechanism}~\citep{dworkOurDataOurselves2006}, 
which simply adds Gaussian noise to the output of a function:
\begin{definition}
    The Gaussian mechanism with noise variance $\sigma_{DP}^2$ and function 
    $f$ outputs $f(X) + \caln(0, \sigma_{DP}^2I)$ for input $X$.
\end{definition}
For a given $(\epsilon, \delta)$-bound and function $f$, the required value 
for $\sigma_{DP}^2$ can be computed tightly using the analytical Gaussian 
mechanism~\citep{balleImprovingGaussianMechanism2018}.

\subsection{Noise-Aware Private Synthetic Data}\label{sec:background-noise-aware-syn-data}
To solve the uncertainty estimation problem for frequentist analyses from 
DP synthetic data, \citet{raisaNoiseAwareStatisticalInference2023}
developed a noise-aware algorithm for generating synthetic data called 
NAPSU-MQ. NAPSU-MQ takes discrete data, summarises it with marginal queries,
releases the query values under DP with the Gaussian mechanism, and finally generates 
multiple synthetic data sets. The downstream analysis is done on each synthetic 
data set, and the results are combined using Rubin's rules for synthetic 
data~\citep{raghunathanMultipleImputationStatistical2003,rubin1993statistical},
which use the multiple analysis results to account for the extra uncertainty 
coming from the synthetic data generation. 

The synthetic data is generated by sampling the posterior predictive distribution 
\begin{equation}
    p(\xpred | \sdp) = \int p(\xpred | \theta)p(\theta | \sdp)\dx \theta,
\end{equation}
where $\theta$ is the parameters of the synthetic data generator and $\sdp$ is 
the noisy marginal query values.
The conditioning on $\sdp$ and including the Gaussian mechanism 
in the model is what makes NAPSU-MQ noise-aware, and allows Rubin's rules to accurately account
for the synthetic data generation and DP noise in the downstream analysis.

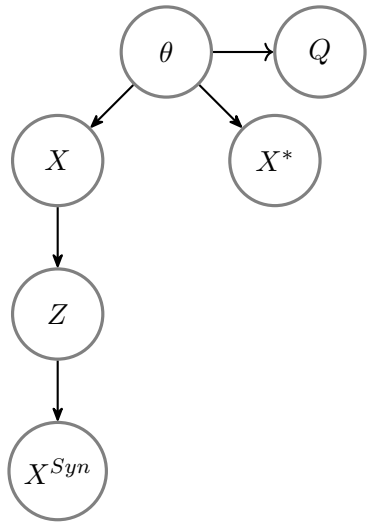
\begin{figure}
    \begin{subfigure}{0.59\textwidth}
        \begin{itemize}
            \item $\theta$: data generating model parameters
            \item $\xreal$: real data
            \item $\xpred$: hypothetical data
            \item $\obs$: observed summary of $X$ ($\obs = X$ without DP)
            \item $\xsyn$: synthetic data, $\xsyn \sim p(\xpred | \obs, I_S)$
            \item $Q$: estimated quantity in downstream analysis
            \item $I_S$: synthetic data provider's background information
            \item $I_A$: analyst's background information
        \end{itemize} 
        \vspace{10mm}
    \end{subfigure}
    \begin{subfigure}{0.4\textwidth}
        \centering
        \begin{tikzpicture}[
            rv/.style={circle, draw=black!50, very thick, minimum size=12mm},
            random edge/.style={thick, >={Stealth[round]}},
            determ edge/.style={thick},
            node distance=8mm,
        ]
            \node[rv] (theta) {$\theta$};
            \node[rv] (Q) [right=of theta] {$Q$}
                edge [<-, determ edge] (theta);
            \node[rv] (y_obs) [below left=of theta] {$\xreal$}
                edge [<-, random edge] (theta);
            \node[rv] (y_star) [below right=of theta] {$\xpred$}
                edge [<-, random edge] (theta);
            \node[rv] (sdp) [below=of y_obs] {$\obs$}
                edge [<-, random edge] (y_obs);
            \node[rv] (xsyn) [below=of sdp] {$\xsyn$}
                edge [<-, random edge] (sdp);
        \end{tikzpicture}  
    \end{subfigure}
    \caption{Left: random variables in noise-aware uncertainty estimation from 
    synthetic data. Right: a Bayesian network describing the dependencies of 
    the random variables. This network is conditional on either $I_S$ or $I_A$,
    depending on whether viewed by the data provider or the analyst.}
    \label{fig:random_variables}
\end{figure}

\section{Bayesian Inference from Synthetic Data}\label{sec:syn-data-bayes-inference}

When the downstream analysis is Bayesian,
the analyst would ultimately want to obtain the 
posterior $p(Q | \xreal, I_A)$ of some quantity $Q$ given real data $X$, where 
$I_A$ denotes the background knowledge such as priors of the 
analyst. We assume the analyst has a method to sample $p(Q | \xreal, I_A)$ if 
they had access to real data, and study what they can do when they only have 
access to synthetic data.

In the DP case, the exact posterior is unobtainable, so we assume that $X$ is only available through a noisy 
summary $\sdp$~\citep{juDataAugmentationMCMC2022,raisaNoiseAwareStatisticalInference2023}, so
the posterior is $p(Q | \sdp, I_A)$.
To unify these notations, we use $\obs$ to denote the observed values, 
so $\obs = \xreal$ in the non-DP case, $\obs = \sdp$ in the DP case, and the 
posterior of interest is $p(Q | \obs, I_A)$. We summarise these random 
variables and their dependencies in Figure~\ref{fig:random_variables}.

In order to introduce the synthetic data into the posterior of interest, we can decompose the posterior as
\begin{equation}
    p(Q | \obs, I_A) = \int p(Q | \obs, \xpred, I_A)p(\xpred | \obs, I_A)\dx \xpred,
    \label{eq:analyst-posterior-decomposition}
\end{equation}
where we abuse notation by using $\xpred$ as the variable to integrate over,
so inside the integral $\xpred$ is not a random variable.
The decomposition in \eqref{eq:analyst-posterior-decomposition} means that we can sample from
$p(Q | \obs, I_A)$ by first sampling the synthetic data from the posterior 
predictive $\xsyn \sim p(\xpred | \obs, I_A)$, and then sampling 
$Q \sim p(Q | \obs, \xpred = \xsyn, I_A)$. The posterior predictive
decomposes as $p(\xpred | \obs, I_A) = \int p(\xpred | \theta, I_A)p(\theta | \obs, I_A)\dx \theta$,
where $\theta$ is the data generating model parameters,
so it is sampled by sampling $\theta \sim p(\theta | \obs, I_A)$, and then sampling 
$\xsyn \sim p(\xpred | \theta, I_A)$.

Note that $\xpred$ and $\xsyn$ are not the same random variable.
$\xpred$ represents a hypothetical real data set that could be obtained if more 
data was collected, as seen in Figure~\ref{fig:random_variables}, and it is not 
the synthetic data set. The synthetic data set $\xsyn$ is a sample from the 
conditional distribution of $\xpred$ given $\obs$. For this reason, 
$p(Q | \obs, \xpred, I_A) \neq p(Q | \obs, I_A)$. To make our notation less 
cluttered, we write $p(\,\cdot\, | \xpred, \,\cdot\,)$ in place of 
$p(\,\cdot\, | \xpred = \xsyn, \,\cdot\,)$ in probabilities when the meaning is 
clear.

There are still two major issues with the decomposition in 
\eqref{eq:analyst-posterior-decomposition}:
\begin{enumerate}
    \item Sampling $p(Q | \obs, \xpred, I_A)$ requires access to $\obs$, which 
    defeats the purpose of using synthetic data.
    \item $\xpred$ needs to be sampled conditionally on the analyst's background 
    information $I_A$, while the synthetic data provider could have different 
    background information $I_S$.
\end{enumerate}
To solve the first issue, in Section~\ref{sec:convergence} we show that if we 
replace $p(Q | \obs, \xpred, I_A)$ inside the integral of \eqref{eq:analyst-posterior-decomposition} with  $p(Q | \xpred, I_A)$,
the resulting distribution converges to the desired posterior,
\begin{equation}
    \int p(Q | \xpred, I_A)p(\xpred | \obs, I_A)\dx \xpred
    \to p(Q | \obs, I_A)
    \label{eq:approx-analyst-posterior-decomposition}
\end{equation}
in total variation distance as the size of $\xpred$ grows. 
It should be noted that many synthetic data sets $\xsyn \sim p(\xpred | \obs, I_A)$ 
will be needed to account for the integral over $\xpred$.

The second issue is known as \emph{congeniality} in the multiple imputation 
literature~\citep{mengMultipleImputationInferencesUncongenial1994,xieDissectingMultipleImputation2016}.
We look at congeniality in the context of Bayesian inference from synthetic data 
in Section~\ref{sec:congeniality}, and find that we can obtain $p(Q | \obs, I_A)$
under appropriate assumptions on the relationship between $I_A$ and $I_S$.

Exactly sampling the LHS of \eqref{eq:approx-analyst-posterior-decomposition}
requires generating a synthetic data set for each sample of $p(Q | \obs, I_A)$, 
which is not practical.
However, we can perform a Monte-Carlo approximation for $p(Q | \obs, I_A)$ by generating $m$ 
synthetic data sets  $\xsyn_1, \dotsc, \xsyn_m \sim p(\xpred | \obs, I_A)$, drawing multiple
samples from each of the $p(Q | \xpred = \xsyn_i, I_A)$,
and mixing these samples, which allows us to obtain more than one sample of 
$p(Q | \obs, I_A)$ per synthetic data set. We look at some properties of this in 
Section~\ref{sec:finite-m}, but we use the integral form in 
\eqref{eq:approx-analyst-posterior-decomposition} in the rest of our theory.

\subsection{Congeniality}\label{sec:congeniality}
In the decomposition \eqref{eq:analyst-posterior-decomposition} of the analyst's 
posterior, $\xpred$ should be sampled conditionally on the 
analyst's background information $I_A$, while in reality the synthetic data 
provider could have different background information $I_S$.

A similar distinction has been studied in the
context of missing data~\citep{mengMultipleImputationInferencesUncongenial1994,xieDissectingMultipleImputation2016},
where the imputer of missing data has a similar role as the 
synthetic data provider. \citet{mengMultipleImputationInferencesUncongenial1994} 
found that Rubin's rules implicitly assume that the 
probability models of both parties are compatible in a certain sense, which 
\citet{mengMultipleImputationInferencesUncongenial1994} 
defined as \emph{congeniality}.

As our examples with Gaussian distributions in Section~\ref{sec:gauss-example} 
show, some notion of congeniality is also 
required in our setting. However, because we study synthetic data instead of 
imputation, and Bayesian instead of frequentist downstream analysis, we need a different 
formal definition. As the analyst only makes inferences on $Q$, it suffices that 
both the analyst and synthetic data provider make the same inferences of $Q$:
\begin{definition}
    The background information sets $I_S$ and $I_A$ are congenial for 
    observation $\obs$ if 
    \begin{equation}
        p(Q | \xpred, I_S) = p(Q | \xpred, I_A)\label{eq:congeniality-x-pred}
    \end{equation}
    for all $\xpred$ and 
    \begin{equation}
        p(Q | \obs, I_S) = p(Q | \obs, I_A).\label{eq:congeniality-obs}
    \end{equation}
\end{definition}
In the non-DP case, \eqref{eq:congeniality-obs} is redundant, as it is implied by 
\eqref{eq:congeniality-x-pred}, but in the DP case, both are needed, as the parties 
may draw different conclusions on $\xreal$ given $\obs = \sdp$.

Combining congeniality and \eqref{eq:approx-analyst-posterior-decomposition}, 
\begin{equation}
    \begin{split}
        \int p(Q | \xpred, I_A)p(\xpred | \obs, I_S)\dx \xpred
        &= \int p(Q | \xpred, I_S)p(\xpred | \obs, I_S)\dx \xpred
        \\&\to p(Q | \obs, I_S) = p(Q | \obs, I_A),
    \end{split}
\end{equation}
where the convergence is in total variation distance as the size of $\xpred$
grows.
In the following, we assume congeniality, and drop $I_A$ and $I_S$ from our notation.

\subsection{Consistency Proof}\label{sec:convergence}
To recap, we want to prove that the posterior from synthetic data,
\begin{equation}
    \pap_n(Q) = \int p(Q | \xpred_n)p(\xpred_n | \obs) \dx \xpred_n, \label{eq:approx-analyst-definition}
\end{equation}
converges in total variation distance to $p(Q | \obs)$ as the size $n$ of 
$\xpred_n$ grows.
We prove this in Theorem~\ref{thm:posterior-approximation}, which
requires that 
both $p(Q | \obs, \xpred_n)$ and $p(Q | \xpred_n)$ approach the same distribution 
as $n$ grows. We formally state this in Condition~\ref{cond:prior-doesnt-matter}.
In Lemma~\ref{thm:bvm-implies-condition}, we show 
that Condition~\ref{cond:prior-doesnt-matter} is a consequence of the 
Bernstein--von Mises theorem (Theorem~\ref{thm:bvm}) under some additional 
assumptions, so we expect it to hold in typical settings.

To make the notation more compact, let $\qex_n \sim p(Q | \obs, \xpred_n)$, 
and let $\qap_n \sim p(Q | \xpred_n)$.

\begin{restatable}{condition}{conditionpriordoesntmatter}\label{cond:prior-doesnt-matter}
    For the observed $Z$ and all $Q$, there exist distributions $D_n$ such that 
    \begin{equation*}
        \TV\left(\qex_n, D_n\right)
        \convprob 0
        \quad \mathrm{and} \quad\,
        \TV\left(\qap_n, D_n\right)
        \convprob 0
    \end{equation*}
    as $n \to \infty$, where the convergence in probability is over sampling 
    $\xpred_n \sim p(\xpred_n | \obs, Q)$.
\end{restatable}

Theorem~\ref{thm:bvm} implies Condition~\ref{cond:prior-doesnt-matter} with some additional 
assumptions:

\begin{restatable}{lemma}{lemmabvmimpliescondition}\label{thm:bvm-implies-condition}
    If the assumptions of Theorem~\ref{thm:bvm} (stated in Condition~\ref{cond:bvm}) 
    hold for the downstream analysis for all $Q_0$, 
    and the following assumptions:
    \begin{enumerate}
        \item[(1)] $\obs$ and $\xpred$ are conditionally independent given $Q$; and
        \item[(2)] $p(\obs | Q) > 0$ for all $Q$,
    \end{enumerate}
    hold, then Condition~\ref{cond:prior-doesnt-matter} holds.
\end{restatable}
\begin{proof}
    The full proof is in Appendix~\ref{sec:missing-proofs-consistency}.
    Proof idea: when $\obs$ and $\xpred$ are conditionally independent given $Q$,
    \begin{equation*}
        p(Q | \obs, \xpred) \propto p(\xpred | Q)p(\obs | Q)p(Q)
    \end{equation*}
    so $p(Q | \obs, \xpred)$ can be equivalently seen as the result of Bayesian inference with 
    observed data $\xpred$ and prior $p(Q | \obs)$. As the only difference to $p(Q | \xpred)$ is 
    the prior, the Bernstein--von Mises theorem implies that both $p(Q | \obs, \xpred)$
    and $p(Q | \xpred)$ converge in total variation distance to the same distribution. 
\end{proof}
Assumption~(1) of Lemma \ref{thm:bvm-implies-condition} will hold if the downstream analysis 
treats its input data as an i.i.d. sample from some distribution. 
Assumption~(2) holds when the likelihood is always 
positive, and in the DP case when the density of the privacy mechanism is positive everywhere, which is 
the case for common DP mechanisms like the Gaussian and Laplace 
mechanisms~\citep{dworkAlgorithmicFoundationsDifferential2014}.

Next is the main theorem of this work: 
\eqref{eq:approx-analyst-posterior-decomposition} holds under 
Condition~\ref{cond:prior-doesnt-matter}.
\begin{restatable}{theorem}{theoremposteriorapproximation}\label{thm:posterior-approximation}
    Under congeniality and Condition~\ref{cond:prior-doesnt-matter}, 
    $\TV\left(p(Q | \obs), \pap_n(Q)\right) \to 0$
    as $n \to \infty$.
\end{restatable}
\begin{proof}
    The full proof is in Appendix~\ref{sec:missing-proofs-consistency}.
    Proof idea: the proof consists of three steps. The first two are in 
    Lemma~\ref{thm:downstream-approximation} and the third is in 
    Lemma~\ref{thm:downstream-approx-to-posterior-approx} in the Appendix.
    The first step is showing that $\TV(\qap_n, \qex_n) \convprob 0$
    when $\xpred_n \sim p(\xpred_n | \obs, Q)$ for fixed $\obs$ and $Q$. This is a simple 
    consequence of the triangle inequality and Condition~\ref{cond:prior-doesnt-matter},
    as total variation distance is a metric. In the second step, we show that 
    $\TV(\qap_n, \qex_n) \convprob 0$ also holds when 
    $\xpred_n \sim p(\xpred_n | \obs)$. In the final step, we show that this 
    implies the claim.
\end{proof}

\subsection{Convergence Rate}
Under a stronger regularity condition, we can get a convergence rate for 
Theorem~\ref{thm:posterior-approximation}. The regularity condition depends on 
uniform integrability:
\begin{restatable}{definition}{definitionuniformintegrability}\label{def:uniform-integrability}
    A sequence of random variables $X_n$ is uniformly integrable if 
    \begin{equation*}
        \lim_{M\to \infty} \sup_n \E(|X_n|\I_{|X_n| > M}) = 0.
    \end{equation*}
\end{restatable}
Now we can state the regularity condition for a convergence rate $O(R_n)$:
\begin{restatable}{condition}{conditionconvergencerate}\label{cond:convergence-rate}
    For the observed $Z$, there exist distributions $D_n$ such that for a sequence 
    $R_1, R_2, \dots > 0$,
    $R_n \rightarrow 0$ as $n \rightarrow \infty$,
    \begin{equation*}
        \frac{1}{R_n} \TV\left(\qex_n, D_n\right)
        \quad \mathrm{and} \quad\,
        \frac{1}{R_n} \TV\left(\qap_n, D_n\right)
    \end{equation*}
    are uniformly integrable when $\xpred_n \sim p(\xpred_n | \obs)$.
\end{restatable}
Note that $\xpred_n \sim p(\xpred_n | \obs)$ conditions on $\obs$, not 
$Q$ and $Z$ like in Condition~\ref{cond:prior-doesnt-matter}. 

Condition~\ref{cond:convergence-rate} is not a standard regularity condition, so 
it is not clear what settings it applies in. To ensure that it at least applies 
in some setting, we show that it is met in univariate Gaussian mean estimation
with $R_n = \frac{1}{\sqrt{n}}$. This is the rate that commonly in the Bernstein--von 
Mises theorem~\citep{hippBernsteinvMisesApproximation1976}.
\begin{restatable}{theorem}{theoremgaussianpriordoesntmatterconvergencerate}\label{thm:gaussian-prior-doesnt-matter-convergence-rate}
    When the up- and downstream models are Gaussian mean 
    estimations with known variance, and $D_n = \caln(\bar{X}, n^{-1}\sigma_k^2)$,
    \begin{equation*}
        \sqrt{n}\TV\big(\qex_n, D_n\big)
        \quad \mathrm{and}\quad\,
        \sqrt{n}\TV\big(\qap_n, D_n\big)
    \end{equation*}
    are uniformly integrable when 
    $\xpred_n \sim p(\xpred_n | \xreal)$.
\end{restatable}
\begin{proof}
    The idea of the proof is to use Pinsker's inequality 
    (Lemma~\ref{thm:total-variation-distance-properties}) to upper bound 
    the total variation distance with KL divergence, and prove the 
    required uniform integrability for the KL divergence upper bound.
    This is a fairly lengthy exercise in upper bounding and computing 
    the limit in the definition of uniform integrability for the various 
    terms that appear in the KL-divergence between the two Gaussians in question.
    We defer the full proof to Appendix~\ref{sec:missing-proofs-convergence-rate}.
\end{proof}

Condition~\ref{cond:convergence-rate} implies an $O(R_n)$ convergence rate:
\begin{restatable}{theorem}{theoremposteriorapproximationconvergencerate}\label{thm:posterior-approximation-convergence-rate}
    Under congeniality and Condition~\ref{cond:convergence-rate},
    $\TV\left(p(Q | \obs), \pap_n(Q)\right) = O(R_n)$.
\end{restatable}
\begin{proof}
    The full proof is in Appendix~\ref{sec:missing-proofs-convergence-rate}.
    Proof idea: first, we prove the uniform integrability of 
    $\frac{1}{R_n}\TV(\qap_n, \qex_n)$ when $\xpred_n \sim p(\xpred_n | \obs)$
    by using the triangle inequality and properties of uniform integrability.
    Second, we prove that this implies the claimed convergence rate.
\end{proof}

\subsection{Finite Number of Synthetic Data Sets}\label{sec:finite-m}

We have now shown that the mixture of posteriors
\begin{equation*}
    \pap_n(Q) = \int p(Q | \xpred_n)p(\xpred_n | \obs) \dx \xpred_n, 
\end{equation*}
converges to the target posterior $p(Q | \obs)$ as $n$ grows. However, sampling 
$\pap_n(Q)$ exactly requires one synthetic data set per sample, 
which is not practical in realistic settings. We can further 
approximate by generating a fixed number $m$ of synthetic datasets and using a 
Monte-Carlo approximation of the integral:
\begin{equation}
    p(Q | \obs) \approx \int p(Q | \xpred_n)p(\xpred_n | \obs)\dx \xpred_n
    \approx \frac{1}{m}\sum_{i=1}^mp(Q | \xpred_n = \xpred_{i,n}),
\end{equation}
with $\xpred_{i,n} \sim p(\xpred_n | \obs)$.

Total variation distance is a metric, so
\begin{equation*}
    \begin{split}
        &\TV\left(\frac{1}{m}\sum_{i=1}^{m} p(Q | \xpred_{i, {n}}), p(Q | \obs)\right)
        \leq\TV\left(\frac{1}{m}\sum_{i=1}^{m} p(Q | \xpred_{i, {n}}), \pap_n(Q) \right) 
        + \TV\left(\pap_n(Q), p(Q | \obs)\right).
    \end{split}
\end{equation*}
Theorem~\ref{thm:posterior-approximation}
gives 
\begin{equation*}
    \lim_{n\to \infty} \TV\left(\pap_n(Q), p(Q | \obs)\right) = 0.
\end{equation*}
If, for all $n$, $p(Q | \xpred_n)$ is continuous for all $\xpred_n$, $p(Q | \xpred_n) \leq h_n(\xpred_n)$ for
an integrable function $h_n(\xpred_n)$, and $\calq \subset \R^d$ is compact, 
the uniform law of large numbers~\citep[Theorem 2]{jennrichAsymptoticPropertiesNonLinear1969} 
gives
\begin{equation}
    \sup_{Q\in \calq} \left|\frac{1}{m}\sum_{i=1}^m p(Q | \xpred_n = \xpred_{i,n}) - \pap_n(Q)\right| \to 0
    \label{eq:finite-m-density-convergence}
\end{equation}
almost surely as $m\to \infty$. If $\calq = \R^d$, we can represent $\calq$ as a countable union of 
compact sets $\calq_k$, apply the uniform law of large numbers on each $\calq_k$, and use the union
bound to obtain \eqref{eq:finite-m-density-convergence} for $\calq$. A similar decomposition of 
$\calq$ can be done for many other constrained parameter sets encountered in practice.

This implies~\citep[Corollary 2.30]{vandervaartAsymptoticStatistics1998}
\begin{align*}
    \lim_{m\to \infty} \TV\left(\frac{1}{m}\sum_{i=1}^{m} p(Q | \xpred_{i, {n}}),
    \pap_n(Q)\right)
    = 0
\end{align*}
for almost all $\xpred_{i, n}$, so 
\begin{equation}
    \lim_{n\to \infty}\lim_{m\to \infty} \TV\left(\frac{1}{m}\sum_{i=1}^{m} p(Q | \xpred_{i, {n}}), 
    p(Q | \obs)\right) = 0
\end{equation}
almost surely when $\xpred_{i, n}\sim p(\xpred_n | \obs)$. 

In Figure~\ref{fig:gaussian-known-variance-hyperparameter-results}, we see that 
the mixture of posteriors becomes very spiky with small $m$ and large $n$, and does 
not fully converge to the target posterior when $n$ grows but $m$ is kept small.
This suggests that
\begin{equation}
    \lim_{m\to \infty}\lim_{n\to \infty} \TV\left(\frac{1}{m}\sum_{i=1}^{m} p(Q | \xpred_{i, {n}}), 
    p(Q | \obs)\right) \neq 0
\end{equation}
because the distributions $p(Q | \xpred_{i, n})$ become narrower 
as $n$ increases, so a fixed number of them is not enough to 
cover $p(Q | \obs)$. This means that in practice, the number of 
synthetic data sets should be increased along with the size of the 
synthetic data sets.

\subsection{Relation to Missing Data Imputation}\label{sec:relation-with-imputation}
\newcommand{\xobs}{X_{obs}}
\newcommand{\xmis}{X_{mis}}

Combining inferences by mixing posteriors from multiple data sets in the style of \eqref{eq:analyst-posterior-decomposition}
was originally proposed by \citet{gelmanBayesianDataAnalysis2004, gelmanBayesianDataAnalysis2014} for Bayesian inference with missing 
data, with completed data sets corresponding to synthetic data sets of our setting. Large completed data sets
are not required in the missing data setting. We explain where this difference between the two 
settings arises next.

In the missing data setting, only a part $\xobs$ of the complete data set
$X$ is observed, while a part $\xmis$ is 
missing~\citep{rubinMultipleImputationNonresponse1987}. To facilitate downstream analysis,
the missing data are imputed by sampling $\xmis \sim p(\xmis | \xobs, I_I)$. Analogously 
with synthetic data, $I_I$ represents the imputer's background knowledge.

Like with synthetic data, we have the decomposition~\citep{gelmanBayesianDataAnalysis2014}
\begin{equation}
    p(Q | \xobs, I_A) = \int p(Q | \xobs, \xmis, I_A)p(\xmis | \xobs, I_A) \dx \xmis.
\end{equation}
If the analyst's and imputer's models are congenial in the sense that 
\begin{equation*}
    p(Q | \xobs, I_A) = p(Q | \xobs, I_I)
\end{equation*}
and 
\begin{equation*}
    p(Q | X, I_A) = p(Q | X, I_I)
\end{equation*}
for any complete data set $X$, then
\begin{equation}
    \begin{split}
        p(Q | \xobs, I_A) 
        = p(Q | \xobs, I_I)
        &= \int p(Q | \xobs, \xmis, I_I)p(\xmis | \xobs, I_I) \dx \xmis
        \\&= \int p(Q | \xobs, \xmis, I_A)p(\xmis | \xobs, I_I) \dx \xmis,
    \end{split}
\end{equation}
so sampling $p(Q | \xobs, I_A)$ can be done by sampling $\xmis \sim p(\xmis | \xobs, I_I)$ multiple times,
sampling $p(Q | \xobs, \xmis, I_A)$ for each $\xmis$, and combining the samples. Unlike with synthetic data,
where sampling $p(Q | \xreal, \xpred, I_A)$ would require the original data and defeat the purpose of 
using synthetic data, sampling $p(Q | \xobs, \xmis, I_A)$ is simply the analysis 
for a complete data set, so generating large imputed data sets is not required.

\section{Non-private Gaussian Examples}\label{sec:gauss-example}
In this section, we look at the Bayesian inference of a univariate Gaussian mean or variance from 
mixing the posteriors from multiple synthetic data sets, which are also generated from the same model.
This allows us to analytically examine various aspects of Bayesian inference from multiple synthetic 
data sets which are not visible in our high-level theory in Section~\ref{sec:syn-data-bayes-inference}, 
as all of the posteriors are analytically tractable and relatively simple. In particular,
we find that the synthetic data set needs to be larger than the original data set 
(Section~\ref{sec:size-of-n-xpred}) and find a variance 
correction formula for the Gaussian setting that gets around this requirement 
(Section~\ref{sec:gaussian-posterior-approximation}). We also look at two forms of uncongeniality,
and find that they cause very different effects: estimating the mean with incorrect known variance
converges to the data provider's posterior, but estimating the variance with incorrect known mean
does not converge to either party's posterior. For reference, we list the posteriors for 
these Gaussian settings in Appendix~\ref{sec:bayes-gauss-background}.

\subsection{Gaussian Mean Estimation with Known Variance}\label{sec:gaussian-known-variance}

Our first example is very simple: $x \sim \caln(\mu, \sigma^2)$, the analyst infers the 
mean $\mu$ of a univariate Gaussian 
distribution with known variance from synthetic data that has been generated 
from the same model.
To differentiate the variables for the analyst and data provider, 
we use bars for the data provider (like $\ups{\sigma}_0^2)$ and hats for the 
analyst (like $\ds{\sigma}_0^2)$.

When the synthetic data is generated from the model with known 
variance $\ups{\sigma}^2_k$,
we sample from the posterior predictive $p(\xpred | \xreal)$ as
\begin{align*}
    \ups{\mu} | \xreal &\sim \caln(\ups{\mu}_{n_\xreal}, \ups{\sigma}_{n_\xreal}^2), \quad
    \xpred | \ups{\mu} \sim \caln^{n_{\xpred}}(\ups{\mu}, \ups{\sigma}_k^2) \\
    \ups{\mu}_{n_\xreal} &= \frac{\frac{1}{\ups{\sigma_0}^2}\ups{\mu}_0 + \frac{n_{\xreal}}{\ups{\sigma}^2_k}\bar{X}}
    {\frac{1}{\ups{\sigma}_0^2} + \frac{n_{\xreal}}{\ups{\sigma}^2_k}}, \quad
    \frac{1}{\ups{\sigma}_{n_\xreal}^2} = \frac{1}{\ups{\sigma}_0^2} + \frac{n_{\xreal}}{\ups{\sigma}_k^2}.
\end{align*}
$\caln^{n_{\xpred}}$ denotes a Gaussian distribution over $n_{\xpred}$ i.i.d. 
samples and $\bar{\xreal}$ is the mean of $\xreal$.

When downstream analysis is the model with known variance $\ds{\sigma}_k^2$, 
we have 
\begin{align*}
    \ds{\mu} | \xpred &\sim \caln(\ds{\mu}_{n_{\xpred}}, \ds{\sigma}_{n_{\xpred}}^2), \quad
    \ds{\mu}_{n_{\xpred}} = \frac{\frac{1}{\ds{\sigma_0}^2}\ds{\mu}_0 + \frac{n_{\xpred}}{\ds{\sigma}^2_k}\bar{X}^*}
    {\frac{1}{\ds{\sigma}_0^2} + \frac{n_{\xpred}}{\ds{\sigma}^2_k}}, \quad
    \frac{1}{\ds{\sigma}_{n_{\xpred}}^2} = \frac{1}{\ds{\sigma}_0^2} + \frac{n_{\xpred}}{\ds{\sigma}_k^2}.
\end{align*}

Now, using $\mu^*$ to denote a sample from the mixture of posteriors
from synthetic data $\pap_n(\mu)$ in \eqref{eq:approx-analyst-definition}, 
we check where the mean and variance of $\mu^*\sim \pap_n(\mu)$ converge when $n_{\xpred} \to \infty$:
\begin{align*}
    \E(\mu^*)  &= \E(\E(\mu^* | \xpred))
    = \E\left(\ds{\mu}_{n_{\xpred}}\right)
    = \E\left(\frac{\frac{1}{\ds{\sigma}_0^2}\ds{\mu}_0 + \frac{n_{\xpred}}
    {\ds{\sigma}^2_k}\xpredbar}{\frac{1}{\ds{\sigma}_0^2} 
    + \frac{n_{\xpred}}{\ds{\sigma}^2_k}}\right)
    \\&= \frac{\frac{1}{\ds{\sigma}_0^2}\ds{\mu}_0 + \frac{n_{\xpred}}
    {\ds{\sigma}^2_k}\E(\xpredbar)}{\frac{1}{\ds{\sigma}_0^2} 
    + \frac{n_{\xpred}}{\ds{\sigma}^2_k}}
    \to \E(\xpredbar) = \ups{\mu}_{n_\xreal}
\end{align*}
as $n_{\xpred} \to \infty$.

For the variance,
\begin{align*}
    \Var(\mu^*) 
    &= \E(\Var(\mu^* | \xpred)) + \Var(\E(\mu^* | \xpred))
    = \E(\ds{\sigma}_{n_{\xpred}}^2) + \Var(\ds{\mu}_{n_{\xpred}}),
\end{align*}
\begin{equation*}
    \E(\ds{\sigma}_{n_{\xpred}}^2)
    = \E\left(\frac{1}{\frac{n_{\xpred}}{\ds{\sigma}_k^2} 
    + \frac{1}{\ds{\sigma}_0^2}}\right) \to 0, n_{\xpred} \to \infty,
\end{equation*}
\begin{equation*}
    \Var\left(\ds{\mu}_{n_{\xpred}}\right)
    = \Var\left(\frac{\frac{n_{\xpred}}{\ds{\sigma}_k^2}\xpredbar 
    + \frac{\ds{\mu}_0}{\ds{\sigma}_0^2}}{\frac{n_{\xpred}}{\ds{\sigma}_k^2} + \frac{1}{\ds{\sigma}_0^2}}\right)
    = \left(\frac{\frac{n_{\xpred}}{\ds{\sigma}_k^2}}
    {\frac{n_{\xpred}}{\ds{\sigma}_k^2} + \frac{1}{\ds{\sigma}_0^2}}\right)^2
    \Var(\xpredbar),
\end{equation*}
and
\begin{align*}
    \Var(\xpredbar) &= \E(\Var(\xpredbar | \ups{\mu}))
    + \Var(\E(\xpredbar | \ups{\mu}))
    \\&= \frac{1}{n_{\xpred}}\E(\Var(x^*_i)) + \Var(\ups{\mu})
    \\&\to \Var(\ups{\mu}) = \ups{\sigma}_{n_\xreal}^2
\end{align*}
as $n_{\xpred} \to \infty$.
Putting these together, 
\begin{align}
    \E(\mu^*) \to \ups{\mu}_{n_\xreal}, \quad
    \Var(\mu^*) \to \ups{\sigma}_{n_\xreal}^2
\end{align}
as $n_{\xpred} \to \infty$.

$\mu^*$ also has a Gaussian distribution, which we will show next.
In 
\begin{equation*}
    \mu^* \sim \int p(\mu | \xpred_n)p(\xpred_n | \xreal)d \xpred,
\end{equation*}
both $p(\mu | \xpred_n) = \mathcal{N}(\hat{\mu}_{n_{\xpred}}, \hat{\sigma}^2_{n_{\xpred}})$ 
and $p(\xpred_n | X)$ are Gaussian.
$\hat{\mu}_{n_{\xpred}}$ is a linear function of $\xpred_n$, so
$p(\hat{\mu}_{n_{\xpred}} | X)$ is also Gaussian. 
$\hat{\sigma}^2_{n_{\xpred}}$ does not depend on $\xpred$, so $\mu^*$ is the 
sum of a random variable with distribution 
$\mathcal{N}(0, \hat{\sigma}^2_{n_{\xpred}})$ and
$\hat{\mu}_{n_{\xpred}}$, which is also Gaussian, meaning that 
$\mu^*$ is Gaussian.

This means that $p(\mu^*) \to p(\mu | \xreal, I_S)$. This is regardless of 
congeniality, which corresponds to both parties having equal known variances
($\ups{\sigma}_k^2 = \ds{\sigma}_k^2)$ in this setting.

\begin{figure}
    \centering
    \includegraphics[width=\textwidth]{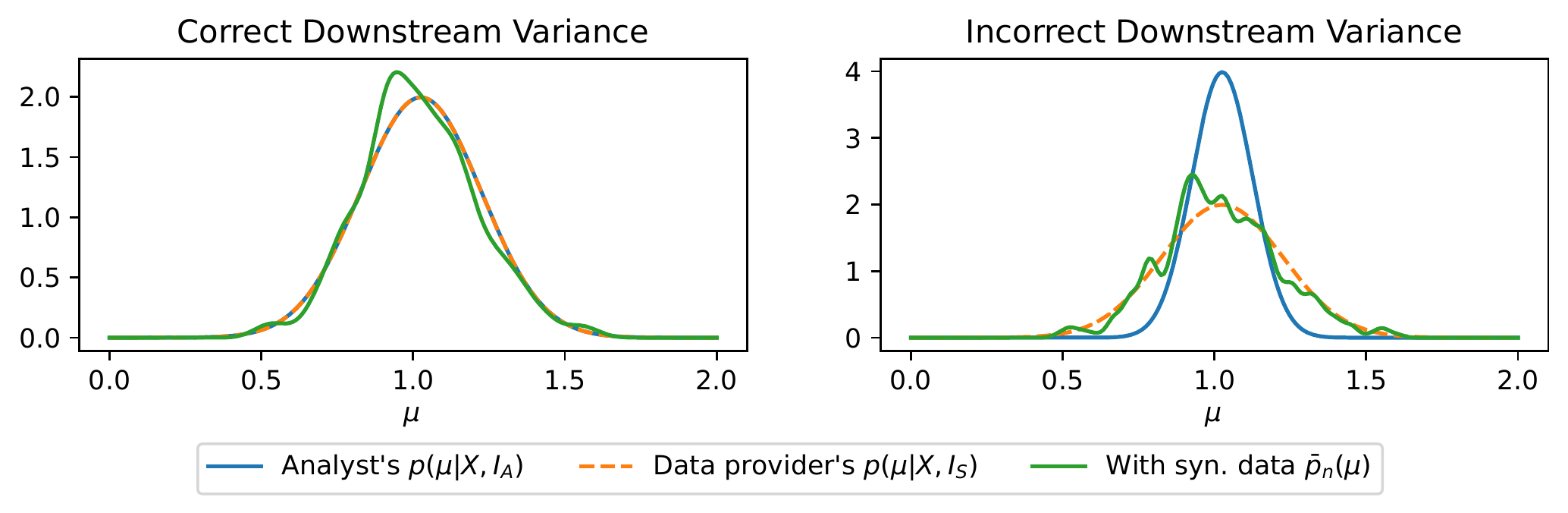}
    \caption{
        Simulation results for the Gaussian mean estimation example, showing that
        the mixture of posteriors from synthetic data in green converges. In the left 
        panel, both the analyst and data provider have the correct known variance.
        The blue and orange lines overlap, as both parties have the same 
        $p(\mu | \xreal)$.
        On the right, the analyst's known variance is too small 
        ($\ds{\sigma}_k^2 = \frac{1}{4}\ups{\sigma}_k^2$), so 
        congeniality is not met, but the mixture of posteriors from synthetic data, $\pap_n(\mu)$, 
        still converges to the data provider's posterior.
        In both panels, $m = 400$ and $\frac{n_{\xpred}}{n_{\xreal}} = 20$.
    }
    \label{fig:known-variance-results}
\end{figure}

We test the theory with a numerical simulation in 
Figure~\ref{fig:known-variance-results}. We generated the real data 
$\xreal$ of size $n_\xreal = 100$ by i.i.d. sampling from $\caln(1, 4)$. 
Both the analyst and data provider use $\caln(0, 10^2)$ as the prior. The 
data provider uses the correct known variance ($\ups{\sigma}_k^2 = 4$),
and the analyst either uses the correct known variance
($\ds{\sigma}_k^2 = 4$), or a too small known variance ($\ds{\sigma}_k^2 = 1$), 
which is an example of uncongeniality.

In the congenial case in the left panel of Figure~\ref{fig:known-variance-results}, both parties 
have the same posterior given the real data 
$X$, and the mixture of posteriors from synthetic data is very close to that. In the 
uncongenial case in the right panel, where the analyst underestimates the variance, the 
parties have different posteriors given $X$, but the mixture of 
synthetic data posteriors is still close to the data provider's posterior.

In Figure~\ref{fig:gaussian-known-variance-hyperparameter-results},
we examine the convergence of the mixture of posteriors from synthetic data under 
congeniality.  We see that setting $n_{\xpred} = n_\xreal$ is not enough, as the 
mixture of posteriors is significantly wider than the analyst's posterior for all values of $m$. 
The synthetic data set needs to be larger than the original, with $n_{\xpred} = 5n_\xreal$ 
already giving a decent approximation and $n_{\xpred} = 20 n_\xreal$ a rather good one with the 
larger values of $m$. We also see that $m$ must be sufficiently large, otherwise 
the method produces very jagged posteriors, for example the top right corner.

\begin{figure}
    \centering
    \includegraphics[width=\textwidth]{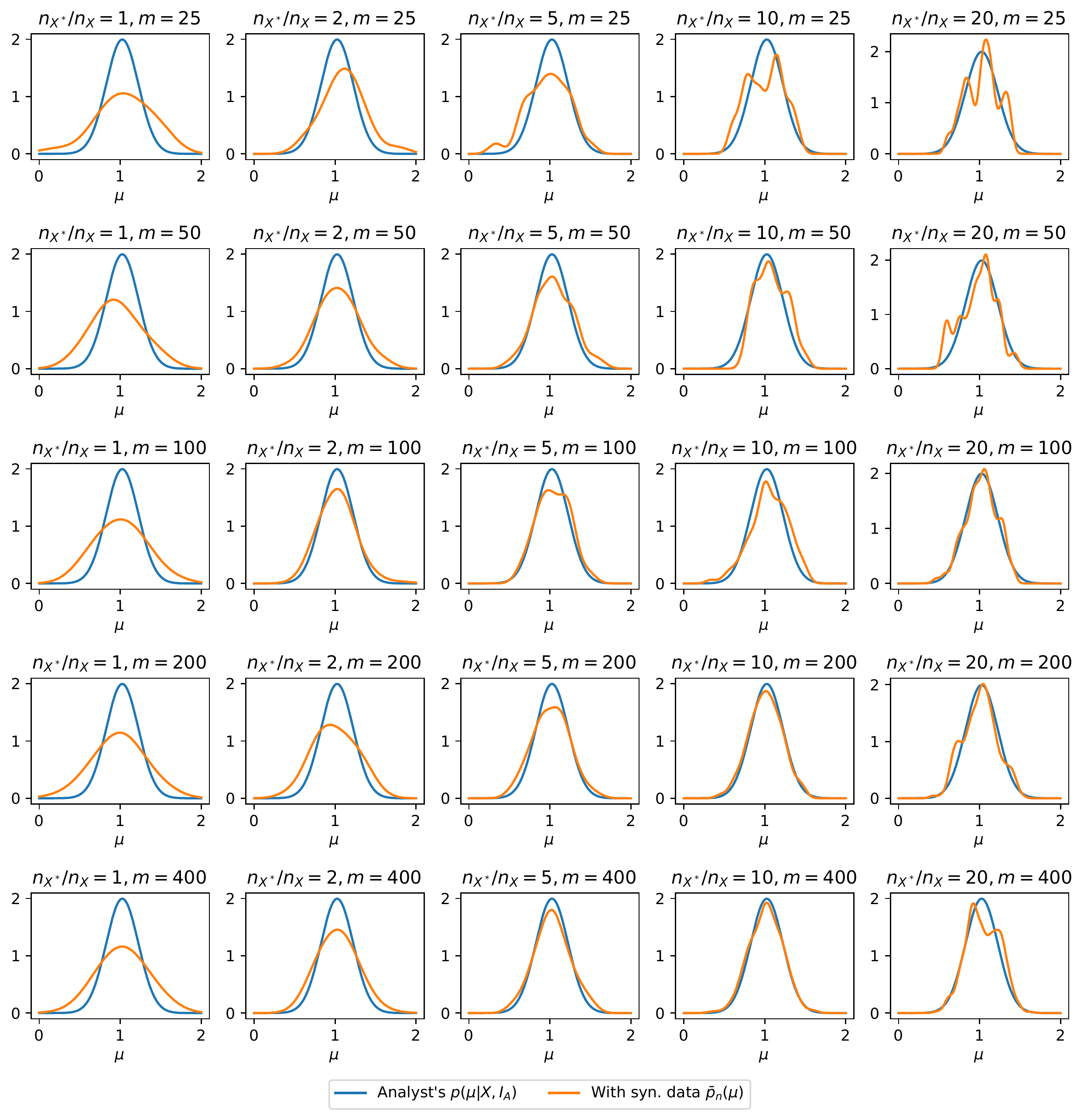}
    \caption{
        Convergence of the mixture of synthetic data posteriors (in orange) 
        with different values of $m$ (increasing top-to-bottom) and $n_{\xpred}$ (increasing left-to-right) in Gaussian mean 
        estimation with known variance. 
    }
    \label{fig:gaussian-known-variance-hyperparameter-results}
\end{figure}

\subsection{Gaussian with Unknown Variance Upstream, Known Variance Downstream}\label{sec:gaussian-unknown-known-variance}

When the synthetic data is generated from the Gaussian mean estimation with unknown 
variance model, $p(\xpred | \xreal)$ is
\begin{align*}    
    \ups{\sigma}^2 | X &\sim \invchisq(\ups{\nu}_{n_\xreal}, \ups{\sigma}_{n_\xreal}^2) \\
    \ups{\mu} | \ups{\sigma}^2, X &\sim \caln\left(\ups{\mu}_{n_\xreal}, \frac{\ups{\sigma}^2}{\ups{\kappa}_{n_\xreal}}\right) \\
    x^*_i | \ups{\mu}, \ups{\sigma}^2 &\sim \caln(\ups{\mu}, \ups{\sigma}^2). \\
\end{align*}

When downstream analysis is the model with known variance $\ds{\sigma}_k^2$, 
$p(\mu^* | \xpred)$ is 
\begin{align*}
    \mu^* | \xpred &\sim \caln(\ds{\mu}_{n_{\xpred}}, \ds{\sigma}_{n_{\xpred}}^2) \\ 
    \ds{\mu}_{n_{\xpred}} &= \frac{\frac{1}{\ds{\sigma_0}^2}\ds{\mu}_0 + \frac{n_{\xpred}}{\ds{\sigma}^2_k}\bar{X}^*}
    {\frac{1}{\ds{\sigma}_0^2} + \frac{n_{\xpred}}{\ds{\sigma}^2_k}} \\
    \frac{1}{\ds{\sigma}_{n_{\xpred}}^2} &= \frac{1}{\ds{\sigma}_0^2} + \frac{n_{\xpred}}{\ds{\sigma}_k^2}.
\end{align*}

Checking where the mean and variance of $\pap_n(\mu)$ converge when $n_{\xpred} \to \infty$:
\begin{align*}
    \E(\mu^*)  &= \E(\E(\mu^* | \xpred))
    = \E\left(\ds{\mu}_{n_{\xpred}}\right)
    = \E\left(\frac{\frac{1}{\ds{\sigma}_0^2}\ds{\mu}_0 + \frac{n_{\xpred}}
    {\ds{\sigma}^2_k}\xpredbar}{\frac{1}{\ds{\sigma}_0^2} 
    + \frac{n_{\xpred}}{\ds{\sigma}^2_k}}\right)
    \\&= \frac{\frac{1}{\ds{\sigma}_0^2}\ds{\mu}_0 + \frac{n_{\xpred}}
    {\ds{\sigma}^2_k}\E(\xpredbar)}{\frac{1}{\ds{\sigma}_0^2} 
    + \frac{n_{\xpred}}{\ds{\sigma}^2_k}}
    \to \E(\xpred) = \ups{\mu}_{n_\xreal}
\end{align*}
as $n_{\xpred} \to \infty$.

For the variance,
\begin{align*}
    \Var(\mu^*) 
    &= \E(\Var(\mu^* | \xpred)) + \Var(\E(\mu^* | \xpred))
    \\&= \E(\ds{\sigma}_{n_{\xpred}}^2) + \Var(\ds{\mu}_{n_{\xpred}}),
\end{align*}
\begin{equation*}
    \E(\ds{\sigma}_{n_{\xpred}}^2)
    = \E\left(\frac{1}{\frac{n_{\xpred}}{\ds{\sigma}_k^2} 
    + \frac{1}{\ds{\sigma}_0^2}}\right) \to 0, n_{\xpred} \to \infty,
\end{equation*}
\begin{equation*}
    \Var\left(\ds{\mu}_{n_{\xpred}}\right)
    = \Var\left(\frac{\frac{n_{\xpred}}{\ds{\sigma}_k^2}\xpredbar 
    + \frac{\ds{\mu}_0}{\ds{\sigma}_0^2}}{\frac{n_{\xpred}}{\ds{\sigma}_k^2} + \frac{1}{\ds{\sigma}_0^2}}\right)
    = \left(\frac{\frac{n_{\xpred}}{\ds{\sigma}_k^2}}
    {\frac{n_{\xpred}}{\ds{\sigma}_k^2} + \frac{1}{\ds{\sigma}_0^2}}\right)^2
    \Var(\xpredbar),
\end{equation*}
and
\begin{align*}
    \Var(\xpredbar) &= \E(\Var(\xpredbar | \ups{\mu}, \ups{\sigma}^2))
    + \Var(\E(\xpredbar | \ups{\mu}, \ups{\sigma}^2))
    \\&= \frac{1}{n_{\xpred}}\E(\ups{\sigma}^2) + \Var(\ups{\mu})
    \\&\to \Var(\ups{\mu}) = \frac{\ups{\sigma}_0^2}{\ups{\kappa}_{n_\xreal}}
\end{align*}
as $n_{\xpred} \to \infty$.
Putting these together, 
\begin{align}
    \E(\mu^*) \to \ups{\mu}_{n_\xreal}, \quad
    \Var(\mu^*) \to \frac{\ups{\sigma}_0^2}{\ups{\kappa}_{n_\xreal}}
\end{align}
as $n_{\xpred} \to \infty$, so $\mu^*$ asymptotically has the same mean and 
variance as the 
marginal posterior $p(\mu | \xreal, I_S)$ of $\mu$ in the synthetic data model, 
which is not the same as the downstream posterior distribution 
$p(\mu | \xreal, I_A)$ on the real data.

We verify this with the simulation in Figure~\ref{fig:unknown-known-variance-results},
where the synthetic data is generated from the model with unknown variance, 
while the analyst uses the known variance model. The setting is otherwise 
identical to the case where both used the known variance model in 
Figure~\ref{fig:known-variance-results}.
The mixture of synthetic data posteriors converges to the data provider's 
posterior, even when the 
analyst uses an incorrect value for the known variance $\ds{\sigma}_k^2$.

\begin{figure}
    \centering
    \includegraphics[width=\textwidth]{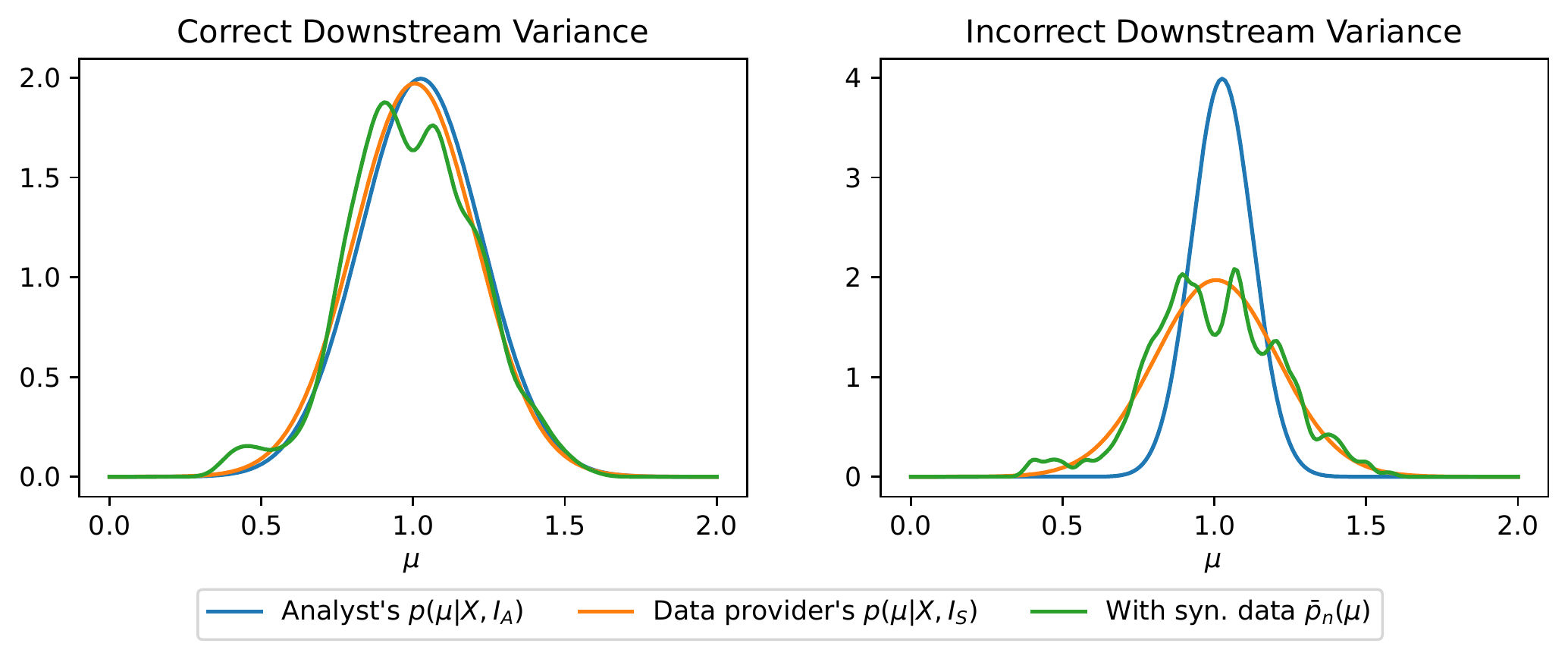}
    \caption{
        Results when the synthetic data is generated from the unknown variance 
        Gaussian mean estimation model, and the analyst uses the model with 
        known variance. On the left, the analyst's known variance is correct,
        on the right it is incorrect. In both cases, the mixture of synthetic 
        data posteriors converges to the data provider's posterior. In both 
        panels, $m = 400$ and $\frac{n_{\xpred}}{n_\xreal} = 20$.
    }
    \label{fig:unknown-known-variance-results}
\end{figure}

\subsection{Size of the Synthetic Data set}\label{sec:size-of-n-xpred}
In the preceding analysis, most of the approximations hold when $n_{\xpred}$
is large, even when $n_{\xpred} \approx n_{\xreal}$. However, based on the 
experiment with different values of $n_{\xpred}$ and $m$ in 
Figure~\ref{fig:gaussian-known-variance-hyperparameter-results},
$n_{\xpred} \gg n_{\xreal}$ is needed for all of the approximations to hold.

This is explained by looking at $\Var(\xpredbar)$. In the case where 
both parties use the known variance model,
\begin{align*}
    \Var(\xpredbar) &= \frac{1}{n_{\xpred}}\E(\Var(x_i^*)) + \Var(\ups{\mu})
    = \frac{1}{n_{\xpred}}(\ups{\sigma}_k^2 + \ups{\sigma}_{n_\xreal}^2)
    + \ups{\sigma}_{n_\xreal}^2
    \\&= \frac{1}{n_{\xpred}}\ups{\sigma}_k^2
    + \left(1 + \frac{1}{n_{\xpred}}\right)
    \frac{1}{\frac{1}{\ups{\sigma}_0^2} + \frac{n_{\xreal}}{\ups{\sigma}_k^2}}.
\end{align*}
If $n_\xreal \approx n_{\xpred}$ and both are large, 
$1 + \frac{1}{n_{\xpred}} \approx 1$ and 
$\frac{1}{\ups{\sigma}_0^2} + \frac{n_{\xreal}}{\ups{\sigma}_k^2} \approx \frac{n_{\xreal}}{\ups{\sigma}_k^2}$, so 
\begin{equation*}
    \Var(\xpredbar) \approx \frac{\ups{\sigma}_{k}^2}{n_{\xpred}} 
    + \frac{\ups{\sigma}_{k}^2}{n_\xreal} 
    \approx \frac{2\ups{\sigma}_{k}^2}{n_\xreal}.
\end{equation*}
With these approximations, 
\begin{equation*}
    \Var(\ups{\mu}) \approx \frac{\ups{\sigma}_{k}^2}{n_\xreal},
\end{equation*}
so 
\begin{equation*}
    \Var(\xpredbar) \approx 2\Var(\ups{\mu}),
\end{equation*}
while the $n_{\xpred} \to \infty$ limit is $\Var(\xpredbar) \to \Var(\ups{\mu})$.
This means that $n_{\xpred} \gg n_{\xreal}$ is required.

The same happens when the synthetic data is generated from the unknown variance 
model:
\begin{align*}
    \Var(\xpredbar) &= \frac{1}{n_{\xpred}}\E(\ups{\sigma}^2) + \Var(\ups{\mu})
    = \frac{1}{n_{\xpred}}
    \frac{\ups{\nu}_0 + n_{\xreal}}{\ups{\nu}_0 + n_\xreal - 2}\ups{\sigma}_{n_{\xreal}}^2
    + \frac{\ups{\sigma}_{n_\xreal}^2}{\ups{\kappa}_0 + n_{\xreal}}.
\end{align*}
If $n_\xreal \approx n_{\xpred}$ and both are large, 
$\frac{\ups{\nu}_0 + n_{\xreal}}{\ups{\nu}_0 + n_\xreal - 2} \approx 1$ and 
$\ups{\kappa}_0 + n_\xreal \approx n_\xreal$, so 
\begin{equation*}
    \Var(\xpredbar) \approx \frac{\ups{\sigma}_{n_\xreal}^2}{n_{\xpred}} 
    + \frac{\ups{\sigma}_{n_\xreal}^2}{n_\xreal} 
    \approx \frac{2\ups{\sigma}_{n_\xreal}^2}{n_\xreal}.
\end{equation*}
With these approximations, 
\begin{equation*}
    \Var(\ups{\mu}) \approx \frac{\ups{\sigma}_{n_\xreal}^2}{n_\xreal},
\end{equation*}
so 
\begin{equation*}
    \Var(\xpredbar) \approx 2\Var(\ups{\mu}).
\end{equation*}

\subsection{Approximate Variance Correction}\label{sec:gaussian-posterior-approximation}
When $n_{\xpred}$ is large, 
\begin{equation*}
    \Var(\mu^*) \approx \E(\ds{\sigma}_{\xpred}^2) + \Var(\bar{X}^*).
\end{equation*}
If $n_{\xpred} = cn_\xreal$ for some $c > 0$, from the analyses in 
Section~\ref{sec:size-of-n-xpred}, we get
\begin{equation*}
    \Var(\bar{X}) \approx \left(1 + \frac{1}{c}\right)\Var(\ups{\mu}),
\end{equation*}
so
\begin{equation*}
    \Var(\mu^*) \approx \E(\ds{\sigma}_{n_{\xpred}}^2) + \left(1 + \frac{1}{c}\right)\Var(\ups{\mu}).
\end{equation*}
Solving for $\Var(\mu)$ gives
\begin{equation}
    \Var(\ups{\mu}) \approx \left(1 + \frac{1}{c}\right)^{-1}\left(\Var(\mu^*) - \E(\ds{\sigma}_{n_{\xpred}}^2)\right).
    \label{eq:gaussian-approximation}
\end{equation}
which gives a Rubin's rules-like~\citep{rubinMultipleImputationNonresponse1987} 
approximation of $\Var(\mu)$ that can be 
computed from smaller synthetic data sets with $n_{\xpred} \approx n_{\xreal}$.

We validate this with the experiment in 
Figure~\ref{fig:gaussian-approximation-results}, which shows that approximating
$p(\mu | X, I_A)$ with a Gaussian with variance from \eqref{eq:gaussian-approximation}
is closer to the real data posterior than 
the mixed posterior approximation from Section~\ref{sec:syn-data-bayes-inference}.

\begin{figure}
    \centering
    \includegraphics[width=\textwidth]{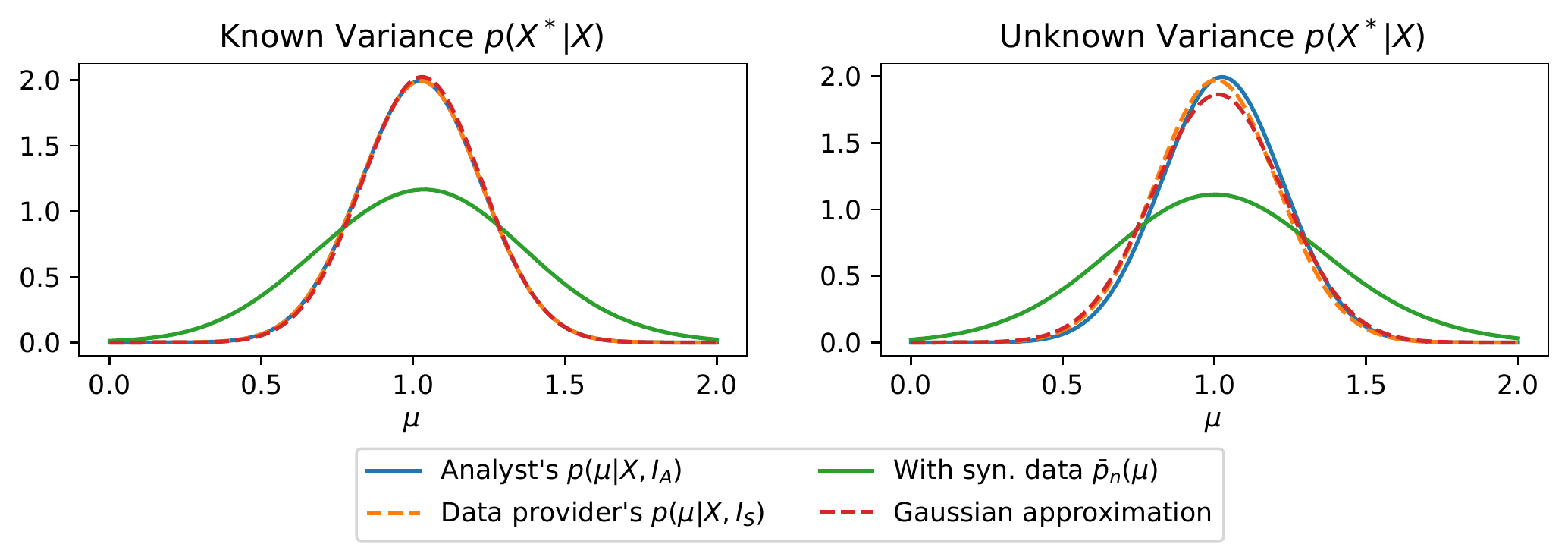}
    \caption{
        Results with the Gaussian approximation with $n_{\xpred} = n_\xreal$,
        showing that the Gaussian approximation is closer to the 
        real data posterior than the mixture of synthetic data posteriors. On 
        the left, the synthetic data is generated from the known variance model, and 
        on the right, the synthetic data is generated from the unknown variance 
        model. In both cases, the known variances for both parties are correct,
        and $m = 400$.
    }
    \label{fig:gaussian-approximation-results}
\end{figure}

\subsection{Gaussian with Known Mean, Unknown Variance}
To asses the effects of uncongeniality when the downstream posterior is not 
Gaussian, we look at Bayesian estimation of the variance of a Gaussian,
with known mean. In this case, the data provider's conjugate prior is 
\begin{equation*}
    \ups{\sigma}^2 \sim \invchisq(\ups{\nu}_0, \ups{\sigma}_0^2),
\end{equation*}
and their known mean is $\ups{\mu}_k$. The synthetic data is generated from
\begin{align*}
    x^*_i|\ups{\sigma}^2 &\sim \caln(\ups{\mu}_k, \ups{\sigma}^2) \\
    \ups{\sigma}^2 | \xreal &\sim \invchisq(\ups{\nu}_{n_\xreal}, \ups{\sigma}_{n_\xreal}^2) \\
    \ups{\sigma}_{n_\xreal}^2 
    &= \frac{\ups{\nu}_0\ups{\sigma}_0^2 + n_\xreal \ups{v}}{\ups{\nu}_0 + n_\xreal} \\
    \ups{\nu}_{n_\xreal} &= \ups{\nu}_0 + n_\xreal \\
    \ups{v} &= \frac{1}{n_\xreal}\sum_{i=1}^{n_\xreal} (x_i - \ups{\mu}_k)^2.
\end{align*}

The analyst's conjugate prior is 
\begin{equation*}
    \ds{\sigma}^2 \sim \invchisq(\ds{\nu}_0, \ds{\sigma}_0^2),
\end{equation*}
their known mean is $\ds{\mu}_k$, and the downstream posterior is 
\begin{align*}
    \ds{v} &= \frac{1}{n_{\xpred}}\sum_{i=1}^{n_{\xpred}} (x_i^* - \ds{\mu}_k)^2 \\
    \ds{\nu}_{n_{\xpred}} &= \ds{\nu}_0 + n_{\xpred} \\
    \ds{\sigma}_{n_{\xpred}}^2 
    &= \frac{\ds{\nu}_0\ds{\sigma}_0^2 + n_{\xpred} \ds{v}}{\ds{\nu}_0 + n_{\xpred}} \\
    \ds{\sigma}^2 | \xpred &\sim \invchisq(\ds{\nu}_{n_{\xpred}}, \ds{\sigma}_{n_{\xpred}}^2).
\end{align*}

Denoting a sample from the mixture of synthetic data posteriors as $\sigma^2_*$,
we have 
\begin{align*}
    \E(\sigma^2_*) &= \E(\E(\sigma^2_* | \xpred))
    = \E\left(\frac{\ds{\nu}_{n_{\xpred}}}{\ds{\nu}_{n_{\xpred}} - 2}\ds{\sigma}^2_{n_{\xpred}}\right)
    = \frac{\ds{\nu}_0 + n_{\xpred}}{\ds{\nu}_0 + n_{\xpred} - 2}
    \E\left(\ds{\sigma}^2_{n_{\xpred}}\right)
    \\&= \frac{\ds{\nu}_0 + n_{\xpred}}{\ds{\nu}_0 + n_{\xpred} - 2}
    \frac{\ds{\nu_0}\ds{\sigma}_0^2 + n_{\xpred} \E(\ds{v})}{\ds{\nu}_0 + n_{\xpred}}
    = \frac{\ds{\nu_0}\ds{\sigma}_0^2 + n_{\xpred} \E(\ds{v})}{\ds{\nu}_0 + n_{\xpred} - 2},
\end{align*}
\begin{align*}
    \E(\ds{v}) &= \frac{1}{n_{\xpred}}\sum_{i=1}^{n_{\xpred}} \E((x_i^* - \ds{\mu}_k)^2)
    = \frac{1}{n_{\xpred}}\sum_{i=1}^{n_{\xpred}} \E((x^*_i)^2 - 2x^*_i\ds{\mu}_k + \ds{\mu}_k^2)
    \\&= \frac{1}{n_{\xpred}}\sum_{i=1}^{n_{\xpred}} \big(\E((x^*_i)^2) - 2\ups{\mu}_k\ds{\mu}_k + \ds{\mu}_k^2\big)
    = \ds{\mu}_k^2 - 2\ups{\mu}_k\ds{\mu}_k + \frac{1}{n_{\xpred}}\sum_{i=1}^{n_{\xpred}} \E((x^*_i)^2)
    \\&= \ds{\mu}_k^2 - 2\ups{\mu}_k\ds{\mu}_k + \frac{1}{n_{\xpred}}\sum_{i=1}^{n_{\xpred}} (\E(x^*_i)^2 + \Var(x_i^*))
    \\&= \ds{\mu}_k^2 - 2\ups{\mu}_k\ds{\mu}_k + \frac{1}{n_{\xpred}}\sum_{i=1}^{n_{\xpred}} (\ups{\mu}_k^2 + \Var(x_i^*))
    \\&= \ups{\mu}_k^2 + \ds{\mu}_k^2 - 2\ups{\mu}_k\ds{\mu}_k + \Var(x_i^*)
    = (\ups{\mu}_k - \ds{\mu}_k)^2 + \Var(x_i^*),
\end{align*}
and
\begin{align*}
    \Var(x_i^*) &= \Var(\E(x_i^* | \ups{\sigma}^2)) + \E(\Var(x_i^* | \ups{\sigma}^2))
    = \E(\ups{\sigma}^2).
\end{align*}
Putting these together, 
\begin{equation}
    \E(\sigma_*^2) \to \E(\ups{\sigma}^2) + (\ups{\mu}_k - \ds{\mu}_k)^2,
    \label{eq:gaussian-known-mean-mean-correction}
\end{equation}
as $n_{\xpred} \to \infty$,
so mixing the downstream posteriors can only recover the data provider's posterior 
when both parties have equal known means.

We verify this with a simulation shown in Figure~\ref{fig:gaussian-known-mean-posterior}.
Both the data provider and analyst use the Gaussian with unknown variance and known 
mean as their model. Otherwise, the setting is identical with the other Gaussian 
examples. When both parties have the correct known mean, $\pap_n(\sigma^2)$ converges 
as expected, but when the analyst has an incorrect known mean, 
$\pap_n(\sigma^2)$ converges to neither party's posterior. However, after applying 
the mean correction from \eqref{eq:gaussian-known-mean-mean-correction},
$\pap_n(\sigma^2)$ appears to have the same variance and shape as the data provider's 
posterior.

\begin{figure}
    \centering
    \includegraphics[width=\textwidth]{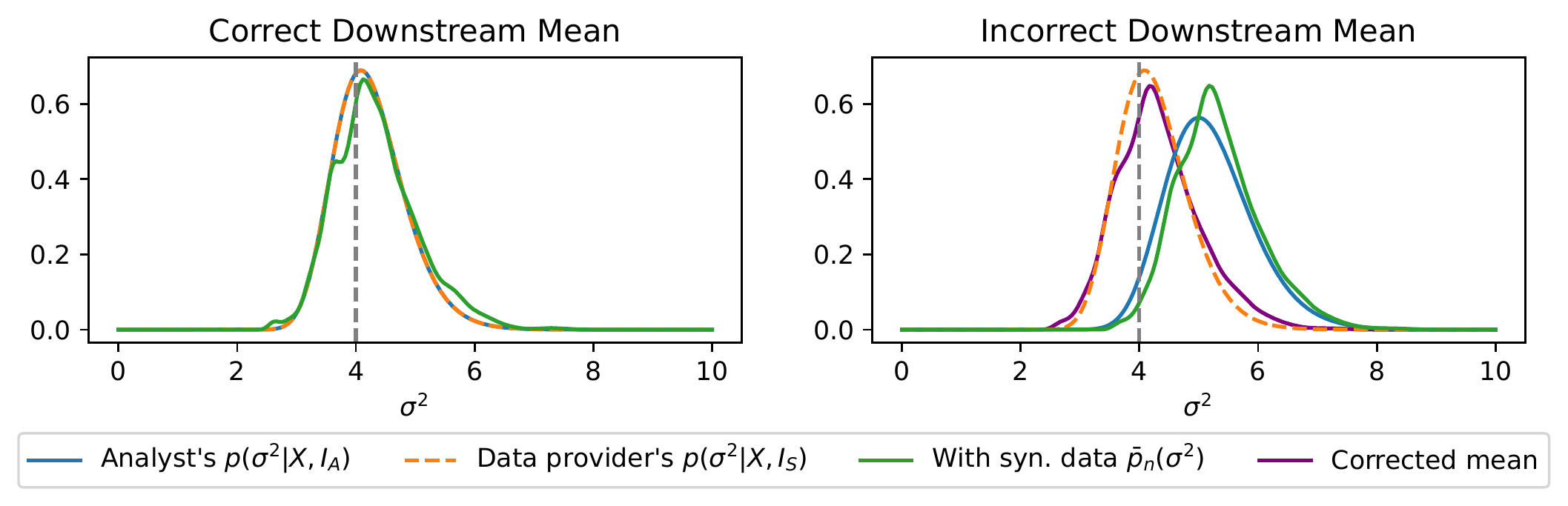}
    \caption{Posteriors from estimating a Gaussian variance $\sigma^2$ with 
    known mean. On the right, both the data provider and the analyst use the 
    correct known mean, so $\pap_n(\sigma^2)$ converges as expected by our theory.
    On the right, the analyst mean is incorrect, so the model is not congenial.
    In this case, $\pap_n(\sigma^2)$ does not converge to either the analyst's or the 
    data provider's posterior. After correcting the mean of $\pap_n(\sigma^2)$ as in 
    \eqref{eq:gaussian-known-mean-mean-correction}, it 
    appears to have the same variance and shape as the data provider's posterior.
    The gray line shows the true parameter value. In both 
    panels, $m = 400$ and $\frac{n_{\xpred}}{n_\xreal} = 20$.   
    }
    \label{fig:gaussian-known-mean-posterior}
\end{figure}

\section{Differentially Private Logistic Regression}\label{sec:logistic-regression-example}

Our second example is logistic regression with DP synthetic data. We consider two 
settings used by \citet{raisaNoiseAwareStatisticalInference2023} with frequentist 
logistic regression, and change the downstream task to Bayesian logistic regression.

Under DP, $\obs$ is a noisy summary $\sdp$ of the real data. 
We need synthetic data sampled from the posterior predictive 
$p(\xpred | \sdp)$, which is 
exactly what the NAPSU-MQ algorithm of 
\citet{raisaNoiseAwareStatisticalInference2023} provides. 
In NAPSU-MQ, 
$\sdp$ contains the values of user-selected marginal queries with added Gaussian 
noise. We used the open-source implementation of 
NAPSU-MQ\footnote{\url{https://github.com/DPBayes/NAPSU-MQ-experiments}} by \citet{raisaNoiseAwareStatisticalInference2023}, and 
describe NAPSU-MQ in Section~\ref{sec:differential-privacy}.

\subsection{Toy Data Logistic Regression}\label{sec:logistic-regression-additional}

The first logistic regression setting we consider uses a simple toy data set
of three binary variables, with $n_\xreal = 2000$ samples. The first 
two variables are sampled with independent coinflips, and the third is 
sampled from logistic regression on the other two, with coefficients 
$(1, 0)$. The prior for the downstream logistic regression is 
$\caln(0, 10I)$.

We generate synthetic data with the 
NAPSU-MQ algorithm~\citep{raisaNoiseAwareStatisticalInference2023}, 
instructing the algorithm to generate $m$ synthetic data sets of size 
$n_{\xpred}$. For the privacy bounds, we vary $\epsilon$, and set 
$\delta = n_\xreal^{-2} \toydatadelta$.

Because of the simplicity of this model, it is possible to use the exact 
posterior decomposition \eqref{eq:analyst-posterior-decomposition} as a baseline,
by using $p(\xreal | \sdp)$ instead of $p(\xpred | \sdp)$ to generate 
synthetic data. We give a detailed description 
of this process in Appendix~\ref{sec:toy-data-sampling-exact-posterior}.
We have also included the DP-GLM 
algorithm~\citep{kulkarniDifferentiallyPrivateBayesian2021} that does not 
use synthetic data, and the non-DP posterior from the real data as 
baselines. We obtained the code for DP-GLM from 
\citet{kulkarniDifferentiallyPrivateBayesian2021} upon request.

\subsubsection{Hyperparameters}
For NAPSU-MQ, we use the hyperparameters of 
\citet{raisaNoiseAwareStatisticalInference2023}, except we used 
NUTS~\citep{hoffmanNoUTurnSamplerAdaptively2014}
with 200 warmup samples and 500 kept samples 
for $\epsilon \in \{0.5, 1\}$,
and 1500 kept samples for $\epsilon = 0.1$,
as the posterior sampling algorithm.
The NAPSU-MQ prior is $\caln(0, 10^2I)$, and the marginal queries are the full set of 3-way 
marginals of all three variables.

The hyperparameters of DP-GLM are the $L_2$-norm upper bound $R$ for the covariates 
of the logistic regression, a coefficient norm upper bound $s$, and the 
parameters of the posterior sampling algorithm DP-GLM uses.
We set $R = \sqrt{2}$ so that the covariates do not get clipped, and set $s = 5$ 
after some preliminary runs. The posterior sampling algorithm is 
NUTS~\citep{hoffmanNoUTurnSamplerAdaptively2014} with 1000 warmup 
samples and 1000 kept samples from 4 parallel chains.

\subsubsection{Results}
Figure~\ref{fig:toy-data-posteriors} compares the mixture of posteriors from synthetic data 
$\pap_n(Q)$ from \eqref{eq:approx-analyst-definition} that uses 
$p(Q | \xpred)$, with $n_{\xpred} / n_\xreal = 20$ and $m = 400$
synthetic data sets, to the baselines. $\pap_n(Q)$ is very close to the 
posterior $p(Q | \sdp)$ from \eqref{eq:analyst-posterior-decomposition}. The 
DP-GLM posterior that does not use synthetic data is somewhat wider.

We ran the experiment 100 times and also with $\epsilon = 0.1$ and 
$\epsilon = 0.5$, and plot coverages and widths of credible 
intervals in Figure~\ref{fig:toy-data-aggregate-results}. 
With $\epsilon = 1$ and $\epsilon = 0.5$,
the coverages are accurate and DP-GLM consistently produces wider intervals.
With $\epsilon = 0.1$, the mixture of synthetic data posteriors likely needs 
more and larger synthetic data sets to converge, as it produced wider and 
slightly overconfident intervals for one coefficient.

Figures~\ref{fig:toy-data-hyperparameter-comparison-tvd} and 
\ref{fig:toy-data-hyperparameter-comparison-tvd-transpose} 
look at how $\pap_n(Q)$ converges to $p(Q | \sdp)$ in terms of total 
variation distance as $n_{\xpred}$ and $m$ increase. They show that increasing 
both together decreases the total variation distance, but just increasing one 
leads to a plateau at some point.

\subsubsection{Plotting Details}
The plotted density of DP-GLM in Figure~\ref{fig:toy-data-posteriors} is 
a kernel density estimate from the posterior samples DP-GLM returns.
The non-DP density is a Laplace approximation. Both synthetic data methods 
use Laplace approximations in the downstream analysis, so their posteriors 
are mixtures of these Laplace approximations for each synthetic data set.
This was also used in Figure~\ref{fig:toy-data-hyperparameter-comparison}.

\begin{figure}
    \centering
    \includegraphics[width=\textwidth]{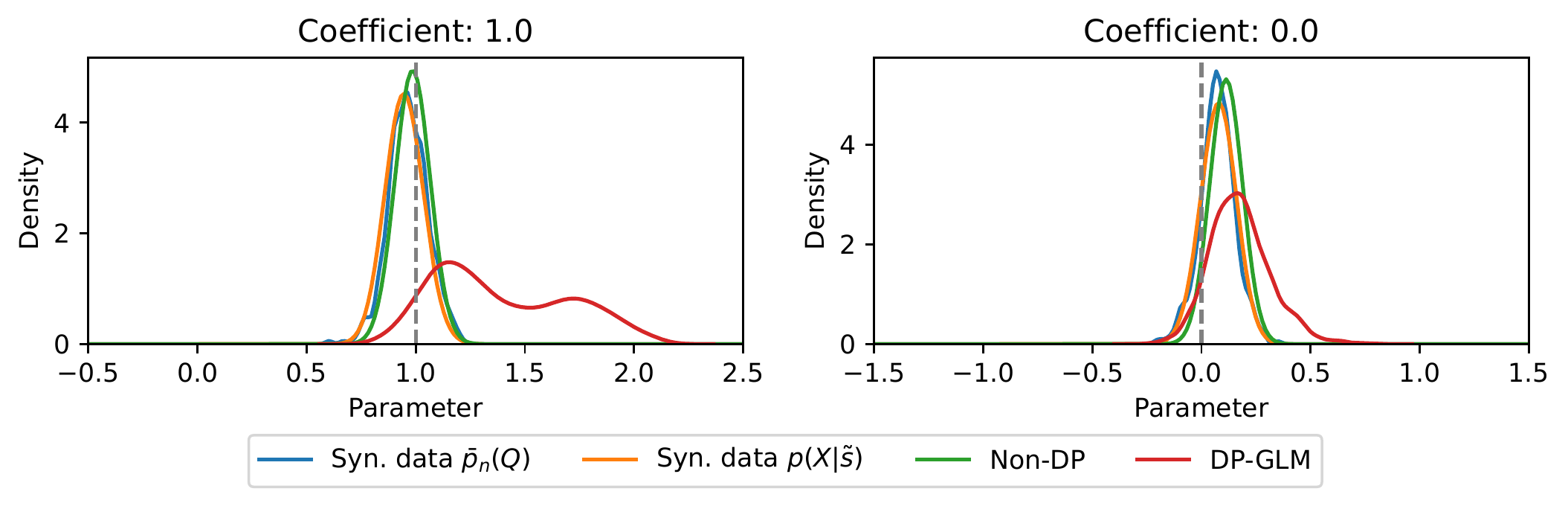}
    \caption{
        Posteriors in the DP logistic regression experiment, where $Q$ are the regression coefficients. The mixture of posteriors from synthetic data, 
        $\pap_n(Q)$, 
        (with $n_{\xpred} / n_\xreal = 20$, $m = 400$) is very close 
        the to the private posterior $p(Q | \sdp)$ computed using \eqref{eq:analyst-posterior-decomposition}. 
        Computing the posterior
        without synthetic data with DP-GLM gives a somewhat wider posterior.
        The true parameter values are highlighted by the grey dashed lines and 
        shown in the panel titles.
        The privacy bounds are $\epsilon = 1$, $\delta = n_\xreal^{-2} 
        \toydatadelta$.
    }
    \label{fig:toy-data-posteriors}
\end{figure}

\begin{figure}
    \centering
    \begin{subfigure}{0.48\textwidth}
        \centering
        \includegraphics[width=\textwidth]{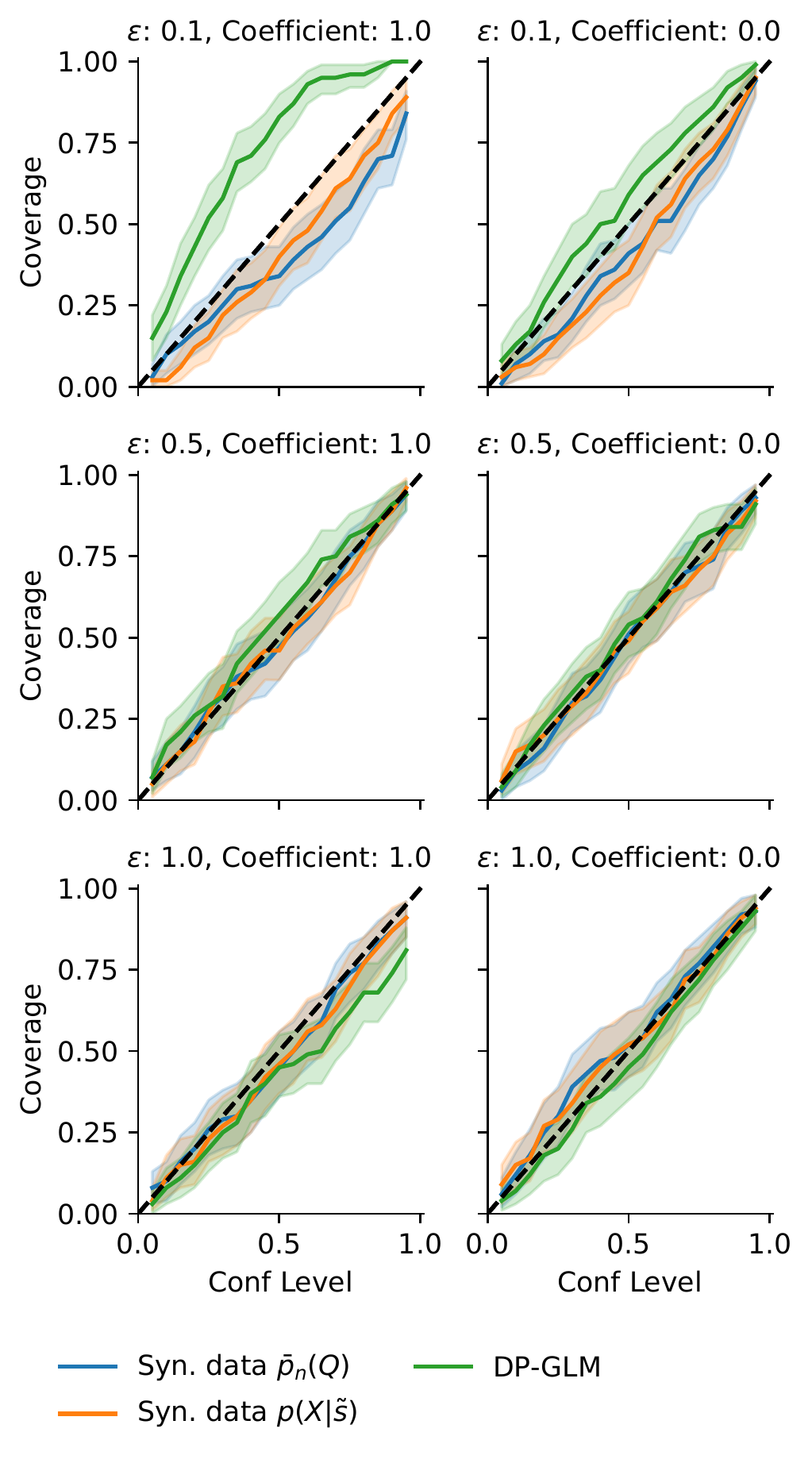}
        \caption{}
        \label{fig:toy-data-coverages}
    \end{subfigure}\quad
    \begin{subfigure}{0.48\textwidth}
        \centering
        \includegraphics[width=\textwidth]{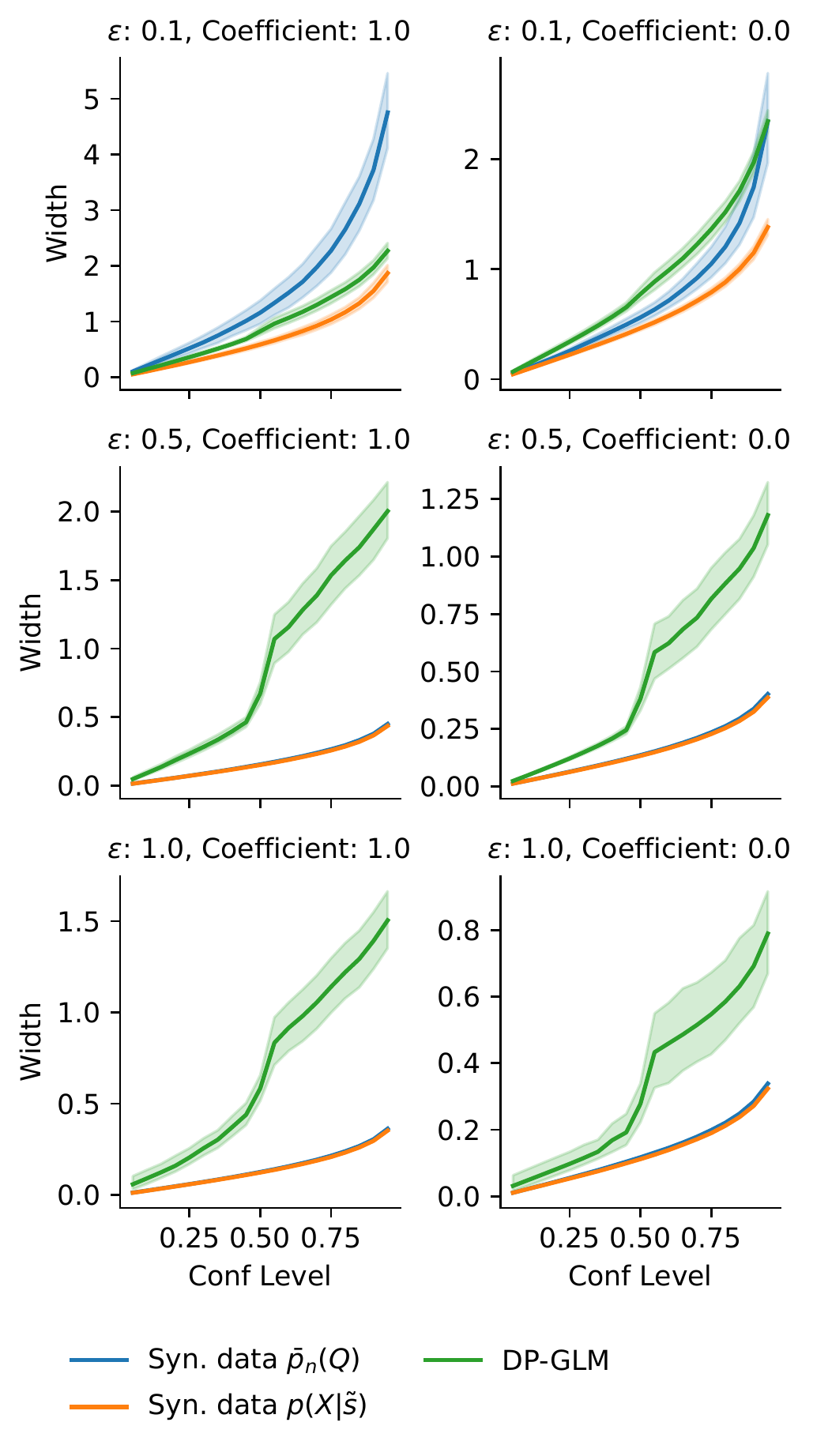}
        \caption{}
        \label{fig:toy-data-widths}
    \end{subfigure}
    \caption{
        (a) Coverages of credible intervals in the toy data experiment.
        The mixture of synthetic data posteriors is accurate, except with 
        $\epsilon = 0.1$, where it may not have converged yet.
        (b) Widths of credible intervals in the toy data experiment.
        DP-GLM produces much wider intervals than other methods, except 
        with $\epsilon = 0.1$.
    }
    \label{fig:toy-data-aggregate-results}
\end{figure}

\begin{figure}
    \centering
    \includegraphics[width=\textwidth]{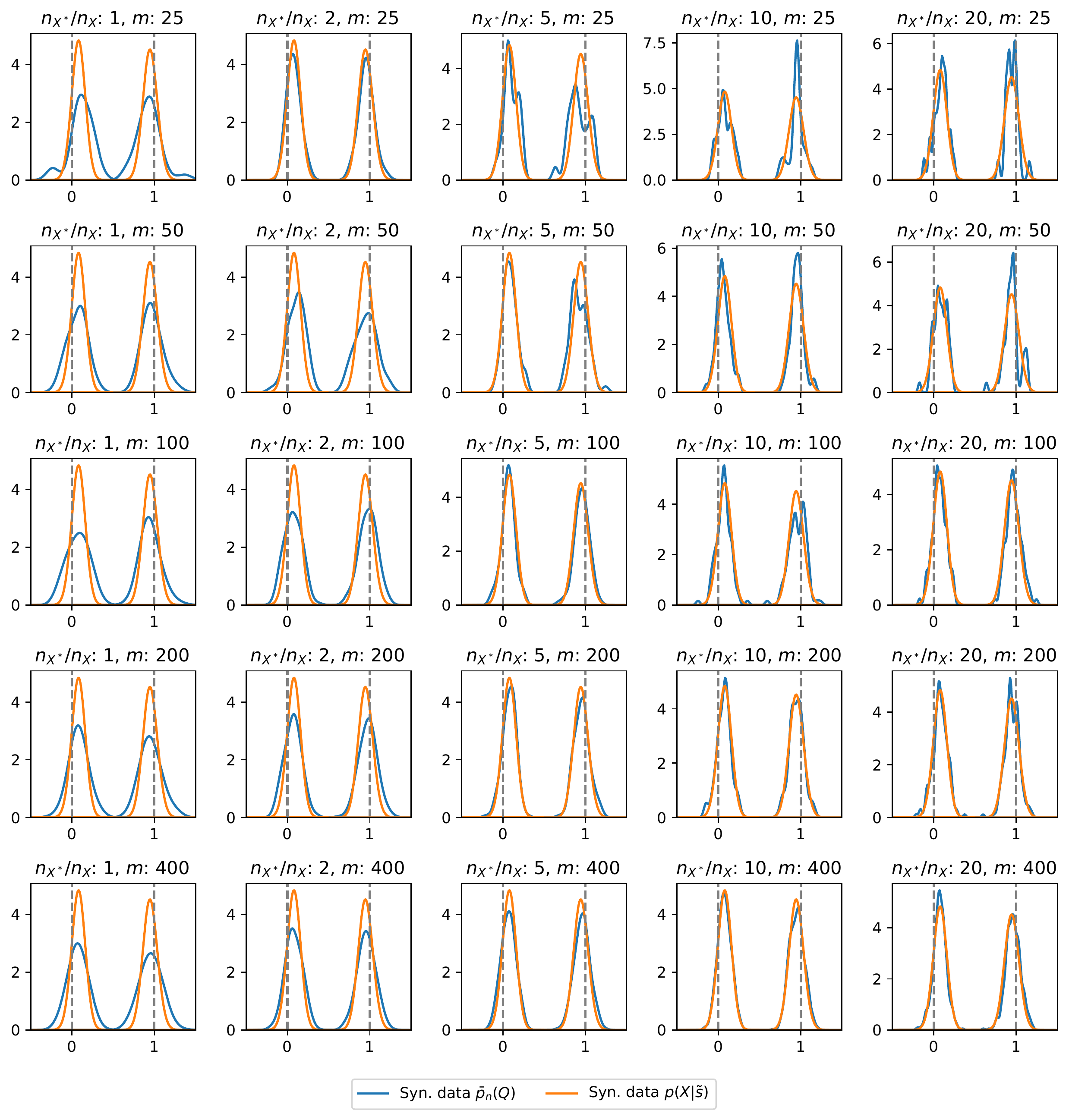}
    \caption{
        Convergence of the mixture of synthetic data posteriors (in blue) with 
        different values of $m$ and $n_{\xpred}$ in the toy data logistic 
        regression experiment.
    }
    \label{fig:toy-data-hyperparameter-comparison}
\end{figure}

\begin{figure}[p]
    \centering
    \includegraphics[width=\textwidth]{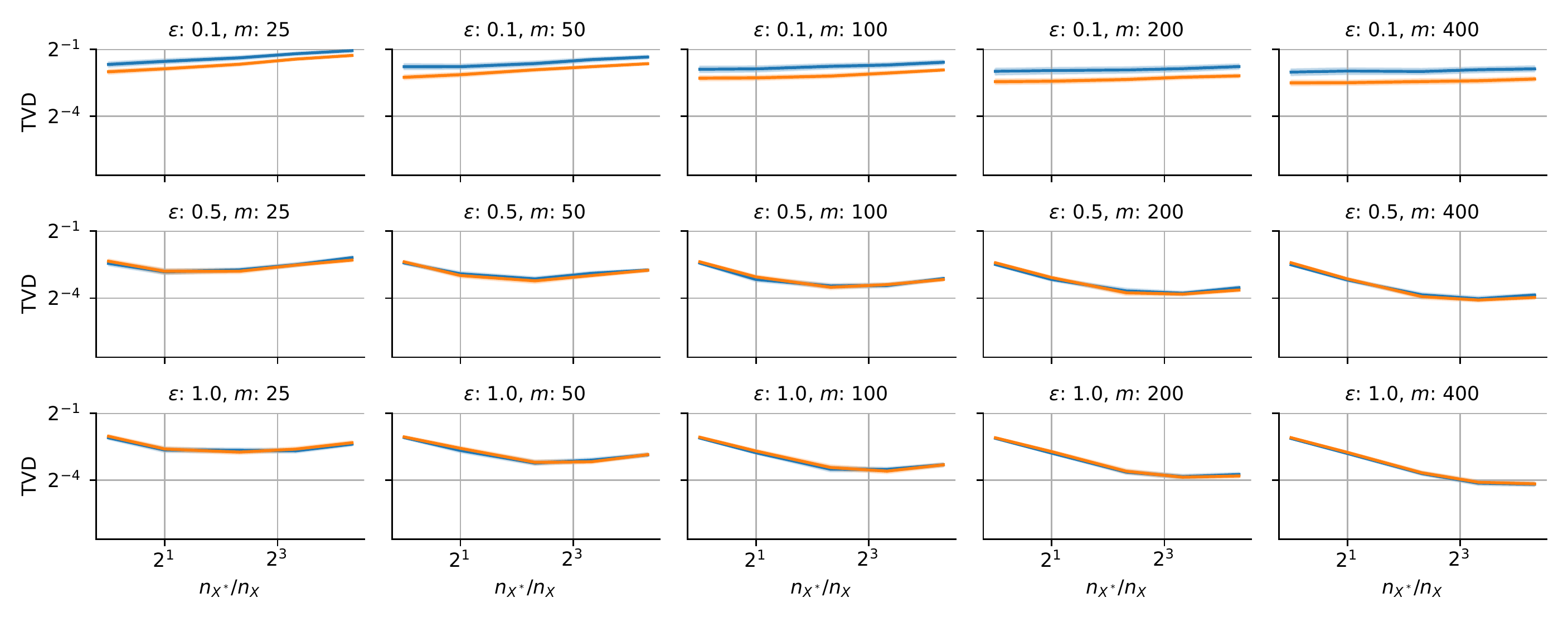}
    \caption{
        Total variation distance (TVD) between $\pap_n(Q)$ and the target 
        $p(Q | \sdp)$ for both 1D marginals (blue and orange)
        in the toy data experiment. For $\epsilon = \{0.5, 1\}$, increasing 
        the size of the synthetic data sets $n_{\xpred}$ decreases the total 
        variation distance at a steady rate,
        until hitting a point where the decrease stops. This point moves 
        further as number of synthetic data sets $m$ increases.
    }
    \label{fig:toy-data-hyperparameter-comparison-tvd}
\end{figure}

\begin{figure}[p]
    \centering
    \includegraphics[width=\textwidth]{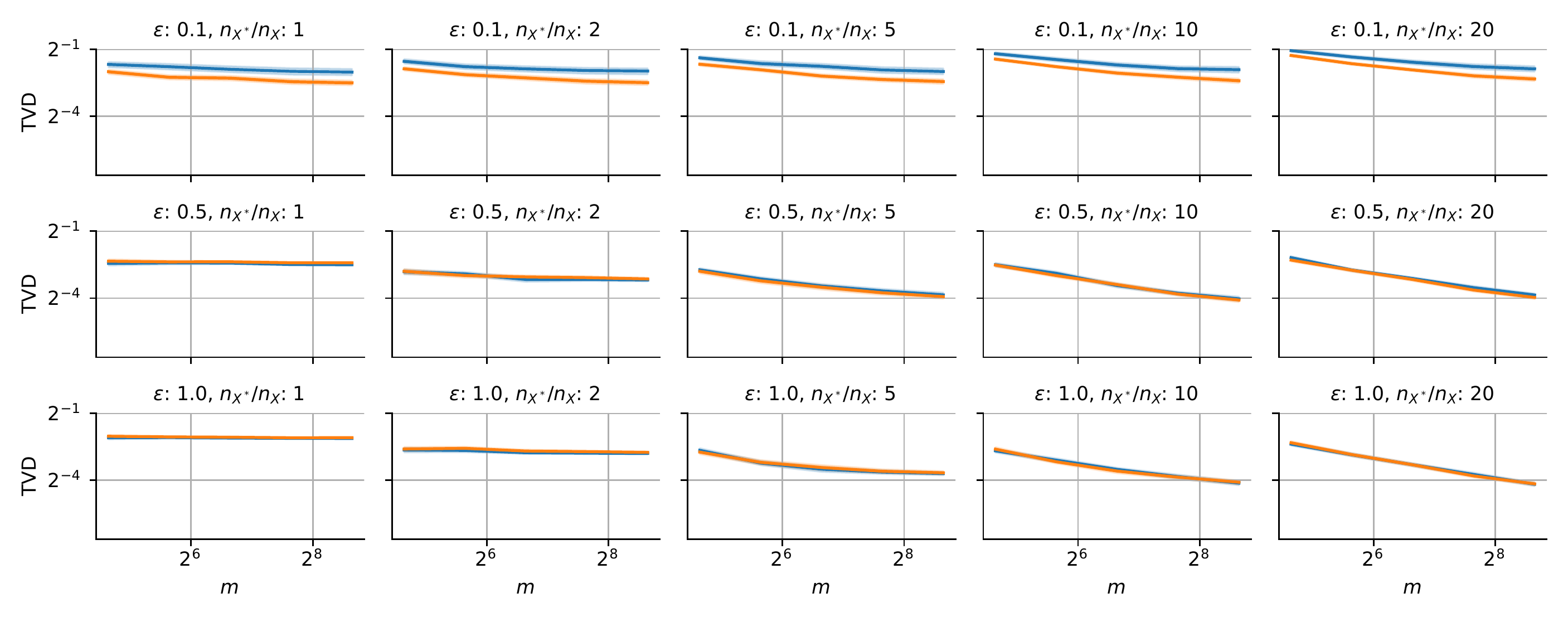}
    \caption{
        Total variation distance (TVD) between $\pap_n(Q)$ and the target 
        $p(Q | \sdp)$ for both 1D marginals (blue and orange)
        in the toy data experiment, with roles of $n_{\xpred}$ and $m$ 
        swapped from Figure~\ref{fig:toy-data-hyperparameter-comparison-tvd}. 
        Increasing $m$ decreases the total variation distance at a rate which
        depends on $n_{\xpred}$ and $\epsilon$.
    }
    \label{fig:toy-data-hyperparameter-comparison-tvd-transpose}
\end{figure}

\subsection{UCI Adult Logistic Regression}\label{sec:adult-experiment}
To test our theory on real data, we used the UCI Adult 
data set~\citep{kohaviAdult1996} setting that was used to test 
NAPSU-MQ~\citep{raisaNoiseAwareStatisticalInference2023}. In this setting, 
the synthetic data set is generated from a subset of 10 columns\footnote{
    age, workclass, education, marital-status, race, gender, capital-gain, 
    capital-loss, hours-per-week and income
}, with the continuous columns age and hours-per-week discretised to 5 
categories, and capital-loss and capital-gain binarised according to whether
they are greater than 0 or not. The income column is already binarised in the 
original data to denote whether it is over $\$50 000$ or not. All rows with 
missing values in the original data set are deleted, which results in 
$n_{\xreal} = 46043$ datapoints. The downstream 
task is logistic regression predicting income using age, race and gender, with 
age converted back to a continuous value by picking the midpoint of each 
category. The reference value for race is ``white'' and for gender is 
``female''. These subsets were originally used to make the runtime of NAPSU-MQ 
manageable, and to make sure that enough relevant information for the downstream 
task can be included in the input queries for 
NAPSU-MQ~\citep{raisaNoiseAwareStatisticalInference2023}.

\subsubsection{Algorithms and Hyperparameters}
The target distribution $p(Q | \obs)$ is not tractable in this setting, so 
we used the non-DP Laplace approximation from the original data set, and the DP 
variational inference (DPVI) 
algorithm~\citep{jalkoDifferentiallyPrivateVariational2017, predigerD3pPythonPackage2022}
as baselines. We also tried running 
DP-GLM~\citep{kulkarniDifferentiallyPrivateBayesian2021}, but we were not 
able to get useful results out of it in this setting. We have also included 
the Gaussian approximation to the mixture of synthetic data posteriors 
discussed in Section~\ref{sec:gaussian-posterior-approximation}, which is 
called ``with variance correction'' in the figures.

The prior for the downstream Bayesian logistic regression is 
$\caln(0, 10I)$, i.i.d. for each coefficient. The privacy parameters are 
$\epsilon \in \{0.25, 0.5, 1\}$, and $\delta = n^{-2} \adultdelta$.
We repeat the experiment 20 times.

The hyperparameters, prior, and selected queries for NAPSU-MQ are the same as in 
the original paper~\citep{raisaNoiseAwareStatisticalInference2023}. 
The synthetic data set size and number are $n_{\xpred} / n_{\xreal} = 10$,
$m = 100$.

DPVI runs DP-SGD~\citep{rajkumarDifferentiallyPrivateStochastic2012, songStochasticGradientDescent2013,abadiDeepLearningDifferential2016}, 
specifically DP-Adam, under the hood, so 
it inherits the clip bound, learning rate, number of iterations, and 
subsampling (without replacement)
ratio hyperparameters from DP-SGD. We tuned these with 
the Optuna library~\citep{akibaOptunaNextgenerationHyperparameter2019a},
using the bounds $[0.1, 50]$ for the clip bound, $[10^{-4}, 10^{-1}]$ for the 
learning rate, $[10^{4}, 10^5]$ for the number of iterations and $[0.001, 1]$ for the 
subsampling ratio. We used the distance of the DPVI posterior mean from the non-DP
real data Laplace approximation as the optimisation criterion. We also tried using 
KL divergence, which gave hyperparameters that produced much wider posteriors.
We used 100 trials for the tuning, and repeated it independently for all 
values of $\epsilon$. The privacy cost of the hyperparameter tuning is 
not reflected in the final results. As the variational posterior, we used a
mean-field Gaussian.

\subsubsection{Results}

Figure~\ref{fig:adult-posterior} compares posteriors from one of the 20 runs
with $\epsilon = 1$. The mixture of synthetic data posteriors $\pap_n(Q)$
is fairly close to the non-DP posterior from the real data set, with the 
exception of two coefficients. The Gaussian approximation to $\pap_n(Q)$ from 
Section~\ref{sec:gaussian-posterior-approximation} is very close to $\pap_n(Q)$.
The posteriors from DPVI are close to non-DP posterior, but for some coefficients,
they are too narrow to overlap the non-DP posterior.

The logistic regression coefficients for which $\pap_n(Q)$ does not work well correspond 
to the two races 
with the smallest number of people in the original data set. These posteriors are 
very wide due to the fact that NAPSU-MQ adds noise uniformly to all queries, 
which means that the queries with small values, corresponding to minority groups 
in the data, get relatively larger amounts of noise.

Figure~\ref{fig:adult-coverage} shows credible interval coverages from the 
Adult experiment, computed from 20 runs. $\pap_n(Q)$ does not achieve perfect 
coverages, as some information is lost due to not running NAPSU-MQ with all 
marginal queries. The coverages are still much better than DPVI.

Figure~\ref{fig:adult-widths} shows the widths of credible intervals 
from the same 20 runs. DPVI produces narrower posteriors than $\pap_n(Q)$,
but the width of $\pap_n(Q)$ posteriors drops as $\epsilon$ increases,
reflecting the reduced uncertainty from DP, which is not the case for DPVI.

In summary, while DPVI is able to find the posterior mean fairly well, it fails
to accurately reflect the additional uncertainty from DP, making the posteriors
overconfident, which would lead to spurious findings if applied in practice.
In contrast, $\pap_n(Q)$ accounts for the DP noise very well, at the cost of 
producing very wide posteriors for the coefficients with a large amount of noise.
Estimating uncertainty reliably is much more important than producing a narrow 
uncertainty estimate: a very wide posterior containing the correct value signals uncertainty, 
but a narrow posterior in the wrong place is confidently incorrect.

\subsubsection{Plotting Details}
The plotted densities for the non-DP posterior and the mixture of synthetic 
data posteriors use Laplace approximations like in the toy data experiment.
The posterior from DPVI is a multivariate Gaussian, which is plotted as is.

\begin{figure}
    \centering
    \includegraphics[width=\textwidth]{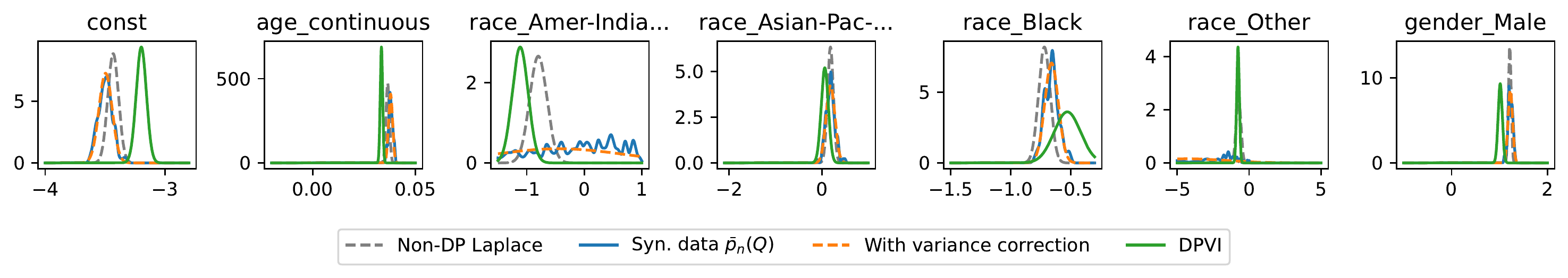}
    \caption{Posteriors from one run on the Adult experiment with $\epsilon = 1$.}
    \label{fig:adult-posterior}
\end{figure}

\begin{figure}
    \centering
    \includegraphics[width=\textwidth]{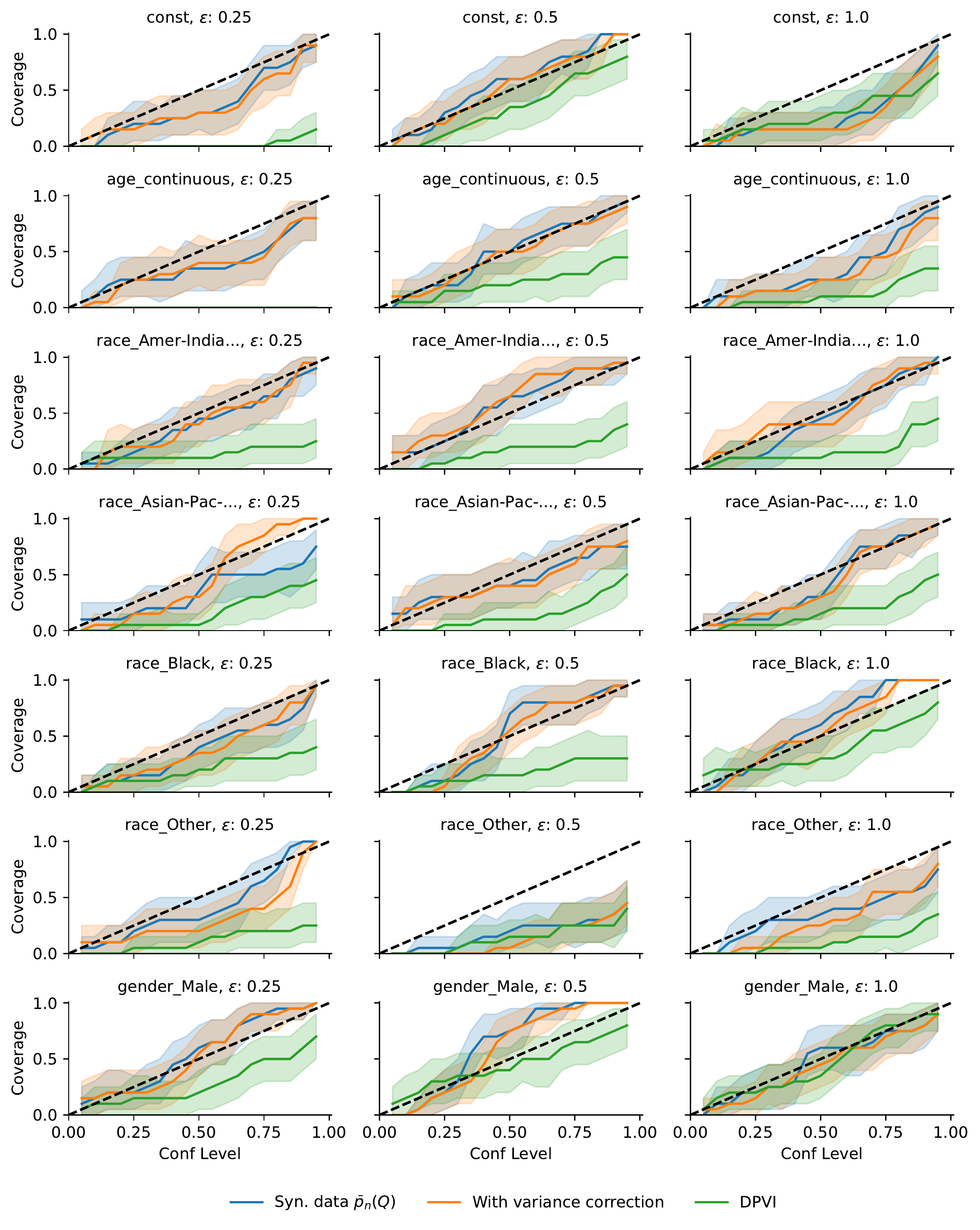}
    \caption{Credible interval coverages on the Adult experiment.}
    \label{fig:adult-coverage}
\end{figure}

\begin{figure}
    \centering
    \includegraphics[width=\textwidth]{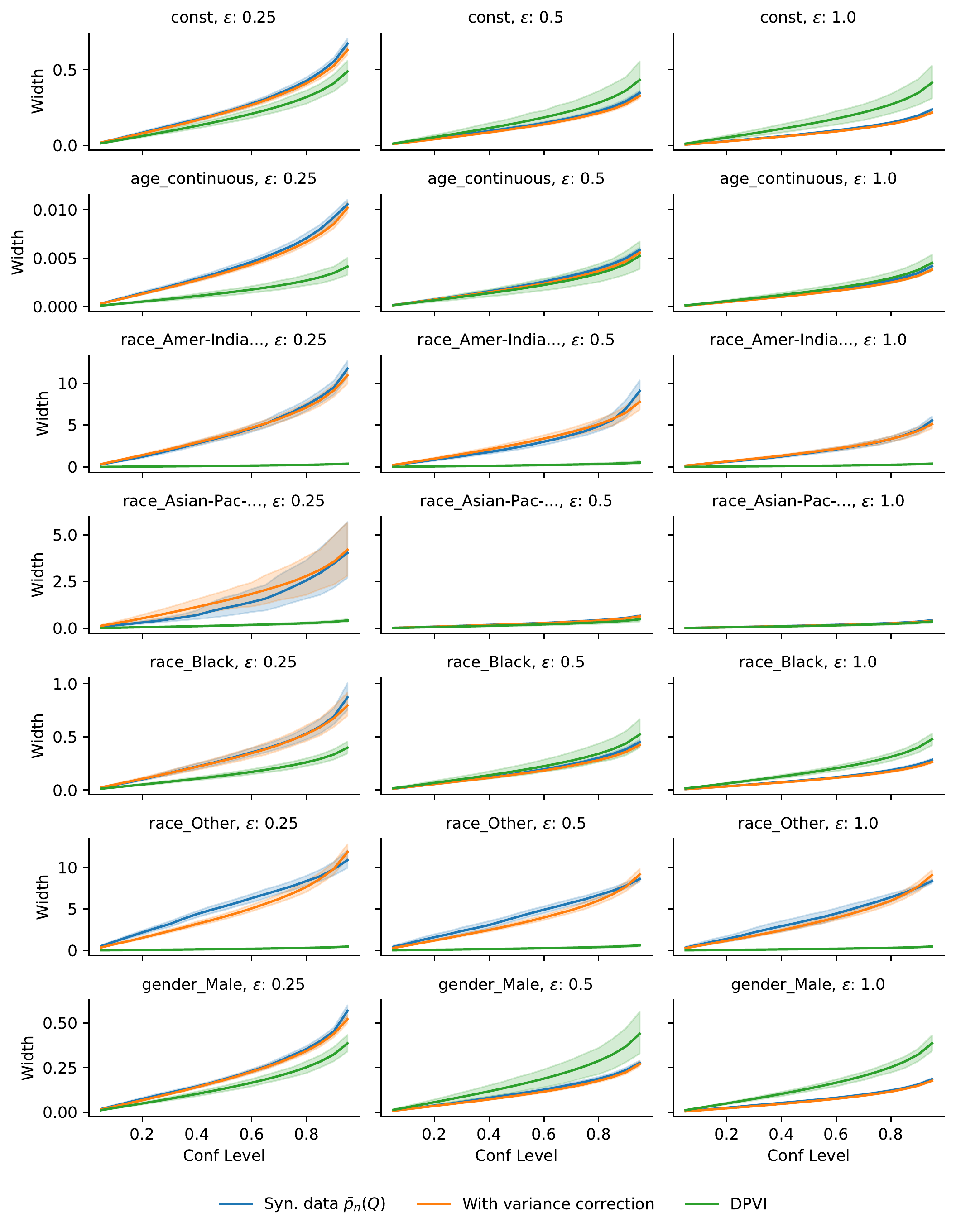}
    \caption{Credible interval widths on the Adult experiment.}
    \label{fig:adult-widths}
\end{figure}

\section{Discussion}
Synthetic data are often considered as a substitute for real data that
are sensitive.
Since the data generation process is based on having access to $Z$, one might
ask why is the synthetic data needed in first place. Why cannot we simply perform
the downstream posterior analysis directly using $Z$? Our analysis allows $Z$ to 
be an arbitrary, even noisy, representation of the data, and it might be difficult
for the analyst to place a model for such generative process for $Q$. In most 
applications, the analyst does have a model for $Q$ arising from the data. Therefore
using the synthetic data as a proxy for the $Z$ allows the analyst to use existing 
models and inference methods to perform the analysis.

\subsection{Limitations}
A clear limitation of mixing posteriors from multiple synthetic data sets is the 
computational cost of analysing many 
large synthetic data sets. This may be substantial for more complex Bayesian 
downstream models, where even a single analysis can be computationally 
expensive. However, the separate analyses can be run in parallel. We also 
expect that the information gained from sampling the posteriors from a few 
synthetic data sets could be used to speed up sampling the others, for 
example by using importance sampling, as they 
likely won't be too far from the sampled ones.

Under DP, we need noise-aware synthetic data generation, which limits 
the settings in which the method can currently be applied. However, if new 
noise-aware methods are developed in the future, the method can 
immediately be used with them.

Condition~\ref{cond:prior-doesnt-matter} limits the applicability of our theory
to downstream analyses where the prior's influence vanishes as the sample size 
grows. This does not always happen for some models, such as some 
infinite-dimensional models, models where the number of parameters increases 
with data set size, and models with a support that heavily depends on the 
parameters. The method also requires congeniality, which 
basically requires the analyst's prior to be compatible with the data provider's. 
Our Gaussian examples in Section~\ref{sec:gauss-example} show that it is sometimes 
possible to recover useful 
inferences even without congeniality. This is not always the case, so an important 
direction for future research is separating these two cases, and 
finding out what can be done in the latter case.

\subsection{Conclusion}
We considered the problem of consistent Bayesian inference 
using multiple, potentially DP, synthetic data sets, and studied an 
inference method 
that mixes the posteriors from multiple large synthetic data sets while re-using 
existing analysis methods designed for real data.
We proved, under congeniality and the general and well-understood regularity conditions of the 
Bernstein--von Mises theorem, that the method is asymptotically exact as 
the sizes of the synthetic data sets grow. 
We studied the method in two examples: non-private Gaussian mean or variance
estimation and 
DP logistic regression. In the former, we were able to use the analytically tractable
structure of the setting to derive additional properties of the method, in particular 
examining what can happen without congeniality.
In both settings, we experimentally validated our theory, and showed that the method works 
in practice. When examining what can go wrong when our assumptions are not met, we found 
that the method can still give sensible results in some cases, but not all, showing that 
the method should be applied with care.
This greatly expands the understanding of Bayesian inference from synthetic data,
filling a major gap in the synthetic data analysis literature.

\acks{
    This work was supported by the Research Council of Finland 
    (Flagship programme: Finnish Center for Artificial Intelligence, 
    FCAI and Grant 356499), the Strategic Research Council 
    at the Research Council of Finland (Grant 358247)
    as well as the European Union (Project
    101070617). Views and opinions expressed are however
    those of the author(s) only and do not necessarily reflect
    those of the European Union or the European Commission. Neither the European 
    Union nor the granting authority can be held responsible for them.
    The authors wish to thank the Finnish Computing Competence 
    Infrastructure (FCCI) for supporting this project with 
    computational and data storage resources.
    We thank Tejas Kulkarni for providing the DP-GLM code.
}

\newpage

\appendix

\section{Additional Background}\label{sec:additional-background}

In this section, we collected background material we use in the rest of the Appendix.
\subsection{Total Variation Distance Properties}\label{sec:total-variation-distance-properties}
Recall the definition of total variation distance:
\definitiontotalvariationdistance*

\begin{lemma}[\citealp{kelbertSurveyDistancesMost2023}]\label{thm:total-variation-distance-properties}
    Properties of total variation distance:
    \begin{enumerate}
        \item For probability densities $p_1$ and $p_2$,
        \begin{equation*}
            \TV(p_1, p_2) = \frac{1}{2}\int |p_1(x) - p_2(x)|\dx x.
        \end{equation*}
        \item Total variation distance is a metric.
        \item Pinsker's inequality: 
        for distributions $P_1$ and $P_2$,
        \begin{equation*}
            \TV(P_1, P_2) \leq \sqrt{\frac{1}{2}\KL{P_1}{P_2}}.
        \end{equation*}
        \item Invariance to bijections: if $f$ is a bijection and $P_1$ and 
        $P_2$ are random variables,
        \begin{equation*}
            \TV(f(P_1), f(P_2)) = \TV(P_1, P_2).
        \end{equation*}
    \end{enumerate}
\end{lemma}
We also occasionally write $\TV(p_1, p_2)$ for probability densities 
$p_1$ and $p_2$ as 
\begin{equation*}
    \TV(p_1, p_2) = \sup_h \left|\int h(x)p_1(x)\dx x - \int h(x)p_2(x)\dx x\right|
\end{equation*}
where $h$ is an indicator function of some measurable set.

\subsection{Bernstein--von Mises Theorem Regularity Conditions}\label{sec:bvm-regularity-conditions}

The version of the Bernstein--von Mises theorem we use is from 
\citet{vandervaartAsymptoticStatistics1998}. To state the regularity 
conditions, we need two definitions:
\begin{definition}
    A parametric probability density $p_Q$ is differentiable in quadratic mean 
    at $Q_0$ if there exists a measurable vector-valued function 
    $\dot{\ell}_{Q_0}$ such that, as $Q \to Q_0$,
    \begin{equation*}
        \int\left(\sqrt{p_Q(x)} - \sqrt{p_{Q_0}(x)}
        - \frac{1}{2}(Q - Q_0)^T \dot{\ell}_{Q_0}(x)\sqrt{p_{Q_0}(x)}\right)^2\dx x
        = o(||Q - Q_0||_2^2).
    \end{equation*}
\end{definition}

\begin{definition}
    A randomised test is a function $\phi\colon \calx \to [0, 1]$.
\end{definition}
The interepretation of $\phi(X)$ is the probability of 
rejecting some null hypothesis after observing data $X$.

Now we can state the regularity conditions of Theorem~\ref{thm:bvm}:
\begin{condition}[\citealp{vandervaartAsymptoticStatistics1998}]\label{cond:bvm}
    For true parameter value $Q_0$ and observed data $X_n$:
    \begin{enumerate}
        \item The datapoints of $X_n$ are i.i.d.
        \item The likelihood $p(x | Q)$ for a single datapoint $x$ is differentiable in quadratic mean at $Q_0$.
        \item The Fisher information matrix of $p(x | Q)$ is nonsingular at $Q_0$.
        \item For every $\beta > 0$, there exists a 
        sequence of randomised tests $\phi_n$ such that 
        \begin{equation*}
            p(X_n | Q_0) \phi_n(X_n) \to 0, \quad 
            \sup_{||Q - Q_0||_2 \geq \beta} p(X_n | Q) (1 - \phi_n(X_n)) \to 0.
        \end{equation*}
        \item The prior $p(Q)$ is absolutely continuous (as a measure) in a 
        neighbourhood of $Q_0$ with a continuous positive density at $Q_0$.
    \end{enumerate}
\end{condition}

\subsection{Bayesian Inference with Gaussian Models}\label{sec:bayes-gauss-background}
In this section, we collect well-known results on Bayesian inference 
of a Gaussian mean. See \citet{gelmanBayesianDataAnalysis2014} for 
proofs.

\subsubsection{Scaled Inverse-Chi-Square Distribution} This parameterisation of 
the inverse gamma distribution is convenient in this setting.
\begin{equation*}
    \invchisq(\nu, s^2) =
    \text{Inv-Gamma}\left(\alpha = \frac{\nu}{2}, \beta = \frac{\nu}{2}s^2\right).
\end{equation*}
If $\theta \sim \invchisq(\nu, s^2)$, $\theta > 0$,
\begin{align*}
    p(\theta) &= \frac{(\frac{\nu}{2})^{\frac{\nu}{2}}}{\Gamma(\frac{\nu}{2})}
    s^\nu \theta^{-(\frac{\nu}{2} + 1)} e^{-\frac{\nu s^2}{2\theta}} \\
    \E(\theta) &= \frac{\nu}{\nu - 2}s^2, \quad \nu > 2 \\ 
    \Var(\theta) &= \frac{2\nu^2}{(\nu - 2)^2(\nu - 4)}s^4, \quad \nu > 4.
\end{align*}

\subsubsection{Gaussian Model with Known Variance}
When the variance of the data is known to be $\sigma_k^2$, and only the mean 
is unknown, the conjugate prior is another Gaussian, and we get the following 
inference problem:
\begin{align*}
    \mu &\sim \caln(\mu_0, \sigma_0^2) \\
    x_i | \mu &\sim \caln(\mu, \sigma_k^2).
\end{align*}
The posterior with $n$ datapoints with sample mean $\bar{X}$ is:
\begin{align*}
    \mu | X &\sim \caln(\mu_n, \sigma_n^2) \\ 
    \mu_n &= \frac{\frac{1}{\sigma_0^2}\mu_0 + \frac{n}{\sigma^2_k}\bar{X}}
    {\frac{1}{\sigma_0^2} + \frac{n}{\sigma^2_k}} \\
    \frac{1}{\sigma_n^2} &= \frac{1}{\sigma_0^2} + \frac{n}{\sigma_k^2}.
\end{align*}

\subsubsection{Gaussian Model with Unknown Variance}
When the variance of the data is also unknown, the conjugate prior is a 
inverse-chi-squared for the variance, and Gaussian for the mean, which gives the 
following inference problem:
\begin{align*}
    \sigma^2 &\sim \invchisq(\nu_0, \sigma_0^2) \\
    \mu | \sigma^2 &\sim \caln\left(\mu_0, \frac{\sigma^2}{\kappa_0}\right) \\
    x_i | \mu, \sigma^2 &\sim \caln(\mu, \sigma^2).
\end{align*}
The joint posterior of $\mu$ and $\sigma^2$ for $n$ datapoints is:
\begin{align*}
    \sigma^2 | X &\sim \invchisq(\nu_n, \sigma_n^2) \\
    \mu | \sigma^2, X &\sim \caln\left(\mu_n, \frac{\sigma^2}{\kappa_n}\right) \\
\end{align*}
with
\begin{align*}
    \bar{X} &= \frac{1}{n}\sum_{i=1}^n x_i \\
    s^2 &= \frac{1}{n - 1}\sum_{i=1}^n (x_i - \bar{X})^2 \\
    \mu_n &= \frac{\kappa_0}{\kappa_0 + n}\mu_0 + \frac{n}{\kappa_0 + n}\bar{X} \\
    \kappa_n &= \kappa_0 + n \\ 
    \nu_n &= \nu_0 + n \\ 
    \nu_n\sigma_n^2 &= \nu_0\sigma^2_0 + (n - 1)s^2 
    + \frac{\kappa_0n}{\kappa_0 + n} (\bar{X} - \mu_0)^2.
\end{align*}
The marginal posterior of $\mu$ is 
\begin{equation*}
    \mu | X \sim t_{\nu_n}\left(\mu_n, \frac{\sigma_n^2}{\kappa_n}\right).
\end{equation*}

\section{Missing Proofs}\label{sec:missing-proofs}

This sections contains the missing proofs from the main text. The proof of our main 
Theorem~\ref{thm:posterior-approximation} is in Appendix~\ref{sec:missing-proofs-consistency},
and proofs of the convergence rate-related Theorems~\ref{sec:missing-proofs-convergence-rate}
and \ref{thm:gaussian-prior-doesnt-matter-convergence-rate} are in 
Appendix~\ref{thm:posterior-approximation-convergence-rate}.
\subsection{Consistency Proof}\label{sec:missing-proofs-consistency}

For ease of reference, we repeat Theorem~\ref{thm:bvm} and 
Condition~\ref{cond:prior-doesnt-matter}:
\theorembvm*

Recall that $\qex_n \sim p(Q | \obs, \xpred_n)$, and $\qap_n \sim p(Q | \xpred_n)$.
\conditionpriordoesntmatter*

\lemmabvmimpliescondition*
\begin{proof}
    Under Assumption~(1)
    \begin{equation*}
        p(Q | \obs, \xpred_n) \propto p(\xpred_n | Q)p(\obs | Q)p(Q),
    \end{equation*}
    so we can view both $p(Q | \obs, \xpred_n)$ and $p(Q | \xpred_n)$ as the posteriors for the 
    same Bayesian inference problem with observed data $\xpred_n$, and priors
    $p(Q | \obs) \propto p(\obs | Q)p(Q)$ and $p(Q)$, respectively. 
    Due to Condition~\ref{cond:bvm} (5) and Assumption~(2), 
    $p(Q | \obs)$ has an everywhere  positive density. Recall that 
    $\qex_n \sim p(Q | \obs, \xpred_n)$ and 
    $\qap_n \sim p(Q | \xpred_n)$. Now, Theorem~\ref{thm:bvm} gives
    \begin{equation*}
        \TV\big(\sqrt{n}(\qex_n - Q_0), \caln(\mu_n, \Sigma)\big) \convprob 0
    \end{equation*}
    and
    \begin{equation*}
        \TV\big(\sqrt{n}(\qap_n - Q_0), \caln(\mu_n, \Sigma)\big) \convprob 0
    \end{equation*}
    as $n \to \infty$, where $\mu_n, \Sigma$ are equal in the two cases because they do not depend on the prior. The probability is over $\xpred_n \sim p(\xpred_n | Q_0)$.
    Because of Assumption (1), $p(\xpred_n | Q_0) = p(\xpred_n | \obs, Q_0)$, so the convergence
    also holds with probability over $\xpred_n \sim p(\xpred_n | \obs, Q_0)$.
    These hold for any $Q_0$.
    Because the function $f_n(q) = \sqrt{n}(q - Q_0)$ is a bijection and 
    total variation distance is invariant to bijections, Condition~\ref{cond:prior-doesnt-matter}
    holds with $D_n$ being the pushforward distribution $D_n = f_n^{-1} \circ \caln(\mu_n, \Sigma)$,
    with the $Q$ of Condition~\ref{cond:prior-doesnt-matter} being $Q_0$.
    Note that $D_n$ is allowed to depend on $Q$ in 
    Condition~\ref{cond:prior-doesnt-matter} due to the order of 
    quantifiers.
\end{proof}

\begin{lemma}\label{thm:downstream-approximation}
    Under Condition~\ref{cond:prior-doesnt-matter}, 
    \begin{equation*}
        \TV(\qex_n, \qap_n) \convprob 0
    \end{equation*}
    as $n \to \infty$, with the probability over $\xpred_n \sim p(\xpred_n | \obs)$.
\end{lemma}
\begin{proof}
    Total variation distance is a metric, so 
    \begin{equation*}
        \TV\big(\qex_n, \qap_n\big)
        \leq 
        \TV\big(\qex_n, D_n\big)
        + \TV\big(\qap_n, D_n\big).
    \end{equation*}
    Now, by Condition~\ref{cond:prior-doesnt-matter}
    \begin{equation}
        \TV\big(\qex_n, \qap_n\big)
        \convprob 0 \label{eq:thm-downstream-approximation-1}
    \end{equation}
    as $n \to \infty$, with the probability over $\xpred_n \sim p(\xpred_n | \obs, Q)$.

    It remains to show \eqref{eq:thm-downstream-approximation-1} with the probability 
    over $\xpred_n \sim p(\xpred_n | \obs)$ instead of $\xpred_n \sim p(\xpred_n | \obs, Q)$.
    With $\xpred_n \sim p(\xpred_n | \obs)$, for any $\epsilon > 0$,
    \begin{align*}
        \Pr_{\xpred_n | \obs}(\TV(\qex_n, \qap_n) > \epsilon)
        &= \int \Pr_{\xpred_n | \obs, Q}(\TV(\qex_n, \qap_n) > \epsilon)p(Q | \obs) \dx Q.
    \end{align*}
    \eqref{eq:thm-downstream-approximation-1} holds for any $Q$, so
    \begin{equation*}
        \lim_{n\to \infty} \Pr_{\xpred_n| \obs, Q}(\TV(\qex_n, \qap_n) > \epsilon) = 0.
    \end{equation*}
    The dominated convergence theorem then implies that 
    \begin{equation*}
        \lim_{n\to \infty} \Pr_{\xpred_n|\obs}(\TV(\qex_n, \qap_n) > \epsilon) = 0,
    \end{equation*}
    so
    \begin{equation*}
        \TV(\qex_n, \qap_n) \convprob 0 
    \end{equation*}
    as $n \to \infty$, with the probability over $\xpred_n \sim p(\xpred_n | \obs)$.
\end{proof}

\begin{lemma}\label{thm:downstream-approx-to-posterior-approx}
    Let $y_n \sim U_n$ be an arbitrary sequence of continuous random variables 
    and let $S(y_n)$, $T(y_n)$ be continuous random variables that depend on $y_n$.
    Let the density functions of $S(y_n)$, $T(y_n)$ and $U_n$ be
    $f_{S(y_n)}, f_{T(y_n)}$ and $f_{U_n}$, respectively. If
    \begin{equation*}
        \TV(S(y_n), T(y_n)) \convprob 0
    \end{equation*}
    as $n \to \infty$, where the probability is over $y_n \sim U_n$, then 
    \begin{equation*}
        \TV\left(\int f_{S(y_n)}(x) f_{U_n}(y_n) \dx y_n, \int f_{T(y_n)}(x) f_{U_n}(y_n) \dx y_n\right)
        \to 0
    \end{equation*}
    as $n \to \infty$.
\end{lemma}
\begin{proof}
    Let $h$ be an indicator function of $x$ over any measurable set and let $\epsilon > 0$. Then
    \begin{align*}
        &\left|\int h(x) \int f_{S(y_n)}(x) f_{U_n}(y_n) \dx y_n \dx x - \int h(x)\int f_{T(y_n)}(x) f_{U_n}(y_n) \dx y_n \dx x\right|
        \\&= \left|\int h(x) \int f_{U_n}(y_n) \big(f_{S(y_n)}(x) - f_{T(y_n)}(x)\big) \dx y_n \dx x\right|
        \\&= \left|\int f_{U_n}(y_n) \int h(x) \big(f_{S(y_n)}(x) - f_{T(y_n)}(x)\big) \dx x \dx y_n\right|
        \\ &\leq \int f_{U_n}(y_n) \left|\int h(x) \big(f_{S(y_n)}(x) - f_{T(y_n)}(x)\big) \dx x \right| \dx y_n
        \\ &= \int f_{U_n}(y_n) \left|\int h(x) f_{S(y_n)}(x) \dx x - \int h(x) f_{T(y_n)}(x) \dx x \right| \dx y_n.
    \end{align*}
    Because $\TV(S(y_n), T(y_n)) \convprob 0$, for large enough $n$, there is a set $Y_n$ with 
    \begin{equation*}
        \TV(S(y_n), T(y_n)) < \frac{\epsilon}{2}
    \end{equation*}
    for all $y_n\in Y_n$, and $\Pr(y_n\in Y_n^C) < \frac{\epsilon}{2}$. As
    \begin{equation*}
        \TV(S(y_n), T(y_n)) = \sup_{h}\left|\int h(x) f_{S(y_n)}(x) \dx x - \int h(x) f_{T(y_n)}(x) \dx x
        \right| \leq 1,
    \end{equation*}
    now 
    \begin{align*}
        &\int f_{U_n}(y_n) \left|\int h(x) f_{S(y_n)}(x) \dx x - \int h(x) f_{T(y_n)}(x) \dx x \right| \dx y_n
        \\
        \begin{split}
        =\int_{Y_n} f_{U_n}(y_n) \left|\int h(x) f_{S(y_n)}(x) \dx x - \int h(x) f_{T(y_n)}(x) \dx x \right| \dx y_n
        \\+\int_{Y_n^C} f_{U_n}(y_n) \left|\int h(x) f_{S(y_n)}(x) \dx x - \int h(x) f_{T(y_n)}(x) \dx x \right| \dx y_n
        \end{split}
        \\&<\int_{Y_n} f_{U_n}(y_n) \frac{\epsilon}{2} \dx y_n + \int_{Y_n^C} f_{U_n}(y_n) \dx y_n
        \\&< \frac{\epsilon}{2} + \frac{\epsilon}{2}
        \\&= \epsilon
    \end{align*}
    for large enough $n$. Now
    \begin{align*}
        &\TV\left(\int f_{S(y_n)}(x) f_{U_n}(y_n) \dx y, \int f_{T(y_n)}(x) f_{U_n}(y_n) \dx y_n\right)
        \\&=  \sup_{h}\left|\int h(x) \int f_{S(y_n)}(x) f_{U_n}(y_n) \dx y_n \dx x 
        - \int h(x)\int f_{T(y_n)}(x) f_{U_n}(y_n) \dx y_n \dx x\right|
        \\&< \epsilon
    \end{align*}
    for any $\epsilon > 0$ with large enough $n$.
\end{proof}

\theoremposteriorapproximation*
\begin{proof}
    The claim follows from Lemma~\ref{thm:downstream-approx-to-posterior-approx} with 
    $y_n = \xpred_n$, $U_n = p(\xpred_n | \obs)$, $S(y_n) \sim p(Q | \xpred_n)$ and 
    $T(y_n) \sim p(Q | \obs, \xpred_n)$.
    These meet the condition for Lemma~\ref{thm:downstream-approx-to-posterior-approx}
    due to Lemma~\ref{thm:downstream-approximation}.
\end{proof}

\subsection{Convergence Rate}\label{sec:missing-proofs-convergence-rate}

\definitionuniformintegrability*

\begin{lemma}\label{thm:uniform-integrable-dominated}
    If $|X_n| \leq Y_n$ almost surely and $Y_n$ is uniformly 
    integrable, $X_n$ is uniformly integrable.
\end{lemma}
\begin{proof}
    \begin{align*}
        0 \leq \lim_{M\to \infty}\sup_n \E(|X_n|\I_{|X_n| > M})
        \leq \lim_{M\to \infty}\sup_n \E(Y_n\I_{Y_n > M})
        = 0.\qedhere
    \end{align*}
\end{proof}

\begin{lemma}[\citealp{billingsleyProbabilityMeasure1995}, Section 16]\label{thm:uniform-integrable-sum}
    If $X_n$ and $Y_n$ are uniformly integrable, $X_n + Y_n$ is uniformly 
    integrable.
\end{lemma}

\conditionconvergencerate*

\theoremposteriorapproximationconvergencerate*
\begin{proof}
    Total variation distance is a metric, so 
    \begin{equation*}
        \frac{1}{R_n}\TV(\qex_n, \qap_n) 
        \leq \frac{1}{R_n} \TV\left(\qex_n, D_n\right)
        + \frac{1}{R_n} \TV\left(\qap_n, D_n\right).
    \end{equation*}
    Now Condition~\ref{cond:convergence-rate} and 
    Lemmas~\ref{thm:uniform-integrable-dominated} and 
    \ref{thm:uniform-integrable-sum} imply that 
    \begin{equation*}
        \frac{1}{R_n}\TV(\qex_n, \qap_n)
    \end{equation*}
    is uniformly integrable with $\xpred_n \sim p(\xpred_n | Z)$.

    Recall that 
    \begin{equation*}
        \frac{1}{R_n}\TV(\qex_n, \qap_n)
        = \frac{1}{R_n}\sup_h\left|\int h(Q)p(Q | \obs, \xpred_n)\dx Q 
        - \int h(Q)p(Q | \xpred_n)\dx Q  \right|
    \end{equation*}
    and 
    \begin{equation*}
        \begin{split}
            &\frac{1}{R_n}\TV(p(Q | \obs), \bar{p}_n(Q))
            \\&= \frac{1}{R_n}\sup_h \left|\int h(Q) \int p(Q | \obs, \xpred_n) p(\xpred_n | \obs) \dx \xpred_n \dx Q
            - \int h(Q) \int p(Q | \xpred_n) p(\xpred_n | \obs) \dx \xpred_n \dx Q\right|,
        \end{split}
    \end{equation*}
    where $h$ is an indicator function of some measurable set.

    For any indicator function $h$, using the start of the proof of Lemma~\ref{thm:downstream-approx-to-posterior-approx}
    gives
    \begin{align*}
        & \frac{1}{R_n}\left|\int h(Q) \int p(Q | \obs, \xpred_n) p(\xpred_n | \obs) \dx \xpred_n \dx Q
        - \int h(Q) \int p(Q | \xpred_n) p(\xpred_n | \obs) \dx \xpred_n \dx Q\right|
        \\&\leq \int p(\xpred_n | \obs) \frac{1}{R_n}\left|\int h(Q)p(Q | \obs, \xpred_n)\dx Q 
        - \int h(Q)p(Q | \xpred_n)\dx Q  \right| \dx \xpred_n
        \\&\leq \int p(\xpred_n | \obs) \frac{1}{R_n}\TV(\qex_n, \qap_n) \dx \xpred_n.
    \end{align*}
    Because $R^{-1}_n\TV(\qex_n, \qap_n)$
    is uniformly integrable when $\xpred_n \sim p(\xpred_n | \obs)$, there exists 
    an $M$ such that for all $n$,
    \begin{equation*}
        \int_{Y_n} p(\xpred_n | \obs) \frac{1}{R_n}\left|\int h(Q) p(Q | \obs, \xpred_n) \dx Q 
        - \int h(Q) p(Q | \xpred_n) \dx Q\right| \dx \xpred_n \leq 1,
    \end{equation*}
    where $Y_n = \{\xpred_n \mid \frac{1}{R_n}\TV(\qex_n, \qap_n) > M\}$.
    Now, for all $n$
    \begin{align*}
        & \frac{1}{R_n}\left|\int h(Q) \int p(Q | \obs, \xpred_n) p(\xpred_n | \obs) \dx \xpred_n \dx Q
        - \int h(Q) \int p(Q | \xpred_n) p(\xpred_n | \obs) \dx \xpred_n \dx Q\right|
        \\&\leq \int p(\xpred_n | \obs) \frac{1}{R_n}\left|\int h(Q)p(Q | \obs, \xpred_n)\dx Q 
        - \int h(Q)p(Q | \xpred_n)\dx Q  \right| \dx \xpred_n
        \\
        \begin{split}
            = \int_{Y_n} p(\xpred_n | \obs) \frac{1}{R_n}\left|\int h(Q)p(Q | \obs, \xpred_n)\dx Q 
            - \int h(Q)p(Q | \xpred_n)\dx Q  \right| \dx \xpred_n
            \\+ \int_{Y_n^C} p(\xpred_n | \obs) \frac{1}{R_n}\left|\int h(Q)p(Q | \obs, \xpred_n)\dx Q 
            - \int h(Q)p(Q | \xpred_n)\dx Q  \right| \dx \xpred_n
        \end{split}
        \\&\leq 1 + \int_{Y_n^C} p(\xpred_n | \obs) M \dx \xpred_n
        \\&\leq 1 + M,
    \end{align*}
    so $\TV\left(p(Q | \obs), \pap_n(Q)\right) = O(R_n)$.
\end{proof}

\theoremgaussianpriordoesntmatterconvergencerate*
\begin{proof}
    When the downstream model is Gaussian mean estimation with 
    known variance, 
    \begin{equation*}
        \qap_n = \caln(\mu_n, \sigma^2_n)
    \end{equation*}
    \begin{equation*}
        \mu_n = \frac{\frac{1}{\sigma_0^2}\mu_0 
        + \frac{n}{\sigma_k^2}\bar{X}}
        {\frac{1}{\sigma_0^2} + \frac{n}{\sigma_k^2}}
    \end{equation*}
    \begin{equation*}
        \frac{1}{\sigma_n^2} = \frac{1}{\sigma_0^2} + \frac{n}{\sigma_k^2}.
    \end{equation*}
    We start with the proof for $\sqrt{n}\TV\big(\qap_n, D_n\big)$.
    By Pinsker's equality and the formula for KL-divergence 
    between Gaussians, 
    \begin{align*}
        \sqrt{n}\TV\big(\qap_n, D_n\big)
        &\leq \sqrt{\frac{1}{2}n\KL{\qap_n}{D_n}}
        \\&=\sqrt{\frac{1}{4}n \left(
        \frac{\sigma_k^2}{n\sigma_n^2}
        + \frac{(\mu_n - \bar{X})^2}{\sigma_n^2}
        -1 + \ln \frac{n\sigma_n^2}{\sigma_k^2}\right)}
        \\&\leq\sqrt{\left|\frac{1}{4}n 
        \left(\frac{\sigma_k^2}{n\sigma_n^2} - 1\right)\right|
        + \frac{1}{4}n \frac{(\mu_n - \bar{X})^2}{\sigma_n^2}
        + \left|\frac{1}{4}n \ln \frac{n\sigma_n^2}{\sigma_k^2}\right|}
        \\&\leq\sqrt{\left|\frac{1}{4}n 
        \left(\frac{\sigma_k^2}{n\sigma_n^2} - 1\right)\right|}
        + \sqrt{\frac{1}{4}n \frac{(\mu_n - \bar{X})^2}{\sigma_n^2}}
        + \sqrt{\left|\frac{1}{4}n \ln \frac{n\sigma_n^2}{\sigma_k^2}\right|}.
    \end{align*}
    The last inequality can be deduced from the fact that the $L_2$-norm is 
    upper bounded by the $L_1$ norm.
    Denote 
    \begin{align*}
        C_1(n) &= n\left(\frac{\sigma_k^2}{n\sigma_n^2} - 1\right)
        = \left(\frac{1}{\sigma_0^2} + \frac{n}{\sigma_k^2}
        \right)\sigma_k^2 - n
        = \frac{\sigma_k^2}{\sigma_0^2} + n - n
        = \frac{\sigma_k^2}{\sigma_0^2}
    \end{align*}
    and
    \begin{align*}
        C_2(n) &= n\ln \frac{n\sigma_n^2}{\sigma_k^2}
        = -n\ln \frac{\sigma_k^2}{n\sigma_n^2}
        = -n\ln \left(\left(\frac{1}{\sigma_0^2} 
        + \frac{n}{\sigma_k^2}\right)\frac{\sigma_k^2}{n}\right)
        \\&= -n\ln \left(\frac{\sigma_k^2}{n\sigma_0^2} + 1\right)
        = -\frac{u\sigma_k^2}{\sigma_0^2}\ln \left(\frac{1}{u} + 1\right)
        = -\frac{\sigma_k^2}{\sigma_0^2}\ln \left(\frac{1}{u} + 1\right)^u,
    \end{align*}
    with $n = \frac{u\sigma_k^2}{\sigma_0^2}$. Because
    \begin{equation*}
        \lim_{u\to \infty}\left(1 + \frac{1}{u}\right)^u = e,
    \end{equation*}
    we have
    \begin{equation*}
        \lim_{u\to \infty}
        -\frac{\sigma_k^2}{\sigma_0^2}\ln \left(\frac{1}{u} + 1\right)^u
        = -\frac{\sigma_k^2}{\sigma_0^2},
    \end{equation*}
    which implies that $C_2(n)$ is bounded.
    
    Futhermore,
    \begin{align*}
        \sqrt{\frac{1}{4}n\frac{(\mu_n - \bar{X})^2}{\sigma_n^2}}
        &= \frac{1}{2}\sqrt{n\left(\frac{1}
        {\sigma_0^2} + \frac{n}{\sigma_k^2}\right)}
        |\mu_n - \bar{X}|.
    \end{align*}
    Denote 
    \begin{equation*}
        s_n = \frac{1}{2}\sqrt{n\left(\frac{1}
        {\sigma_0^2} + \frac{n}{\sigma_k^2}\right)}.
    \end{equation*}
    Note that $s_n = O(n)$. Then 
    \begin{align*}
        \sqrt{\frac{1}{4}n\frac{n(\mu_n - \bar{X})^2}{n\sigma_n^2}}
        = s_n|\mu_n - \bar{X}|,
    \end{align*}
    and
    \begin{align*}
        s_n|\mu_n - \bar{X}|
        &= s_n\left|\frac{\frac{1}{\sigma_0^2}\mu_0}{\frac{1}{\sigma_0^2} + \frac{n}{\sigma_k^2}}
        + \frac{\frac{n}{\sigma_k^2}}{\frac{1}{\sigma_0^2} + \frac{n}{\sigma_k^2}}\bar{X} - \bar{X} \right|
        \\&\leq s_n\frac{\frac{1}{\sigma_0^2}\mu_0}{\frac{1}{\sigma_0^2} + \frac{n}{\sigma_k^2}}
        + s_n\left|\frac{\frac{n}{\sigma_k^2}}{\frac{1}{\sigma_0^2} + \frac{n}{\sigma_k^2}} - 1\right||\bar{X}|
        \\&= s_n\frac{\frac{1}{\sigma_0^2}\mu_0}{\frac{1}{\sigma_0^2} + \frac{n}{\sigma_k^2}}
        + s_n\left|\frac{\frac{n}{\sigma_k^2}}{\frac{1}{\sigma_0^2} + \frac{n}{\sigma_k^2}} - 
        \frac{\frac{1}{\sigma_0^2} + \frac{n}{\sigma_k^2}}{\frac{1}{\sigma_0^2} + \frac{n}{\sigma_k^2}}\right||\bar{X}|
        \\&= s_n\frac{\frac{1}{\sigma_0^2}\mu_0}{\frac{1}{\sigma_0^2} + \frac{n}{\sigma_k^2}}
        + s_n\frac{\frac{1}{\sigma_0^2}}{\frac{1}{\sigma_0^2} + \frac{n}{\sigma_k^2}}|\bar{X}|.
    \end{align*}
    
    Denote 
    \begin{equation*}
        C_3(n) = s_n\frac{\frac{1}{\sigma_0^2}\mu_0}{\frac{1}{\sigma_0^2} + \frac{n}{\sigma_k^2}},
    \end{equation*}
    and 
    \begin{equation*}
        C_4(n) = s_n\frac{\frac{1}{\sigma_0^2}}{\frac{1}{\sigma_0^2} + \frac{n}{\sigma_k^2}}.
    \end{equation*}
    Because $s_n = O(n)$, we have $C_3(n) = O(1)$ and $C_4(n) = O(1)$, so $C_3(n)$ and $C_4(n)$ are 
    bounded.
    
    We now have
    \begin{align*}
        \sqrt{n}\TV\big(\qap_n, D_n\big)
        \leq \sqrt{\frac{1}{4}|C_1(n)|} + \sqrt{\frac{1}{4}|C_2(n)|}
        + C_3(n) + C_4(n)\bar{X}.
    \end{align*}
    By Lemmas~\ref{thm:uniform-integrable-dominated} and \ref{thm:uniform-integrable-sum}, 
    it suffices to show that each of the terms on the right is uniformly integrable.
    The terms containing $C_1, C_2$ and $C_3$ are non-random and bounded in $n$, 
    so they are uniformly integrable. It remains to show that 
    $C_4(n)\bar{X}$ is uniformly integrable.  $C_4(n)$ is bounded,
    so we only need to show that $\bar{X}$ is uniformly integrable.
    
    To bound the expectation in the definition of uniform integrability for 
    $|\bar{X}|$, we need some background facts.
    For geometric series, with $a\in \R$ and $|r| < 1$, 
    \begin{equation*}
        \sum_{i=0}^\infty ar^i = \frac{a}{1 - r},
    \end{equation*}
    and differentiating both sides with regards to $r$ gives 
    \begin{equation*}
        \sum_{i=0}^\infty a(i+1)r^{i} = \frac{a}{(1 - r)^2}.
    \end{equation*}
    For a Gaussian random variable $Y$ with mean $\mu$ and variance $\sigma$,
    $\Pr(Y > \mu + t) \leq e^{-\frac{t^2}{2\sigma^2}}$.
    $\bar{X} \sim \caln\left(\mu, \frac{1}{n}\sigma_k^2\right)$, so this tail bound 
    gives
    \begin{equation*}
        \Pr(\bar{X} > t + \mu) \leq 2e^{-\frac{nt^2}{2\sigma_k^2}}, \quad
        \Pr(\bar{X} < \mu - t) \leq 2e^{-\frac{nt^2}{2\sigma_k^2}}.
    \end{equation*}
    Now
    \begin{align*}
        &\lim_{M\to \infty}\sup_n 
        \E\big(|\bar{X}|\I_{|\bar{X}| > M}\big)
        \\&= \lim_{M\to \infty}\sup_n 
        \E_\mu\big(\E(|\bar{X}|\I_{|\bar{X}| > M} | \mu)\big)
        \\&= \lim_{M\to \infty}\sup_n 
        \E_\mu\left(
        \sum_{i=0}^\infty \E\big(|\bar{X}|
        \I_{M + i < |\bar{X}| \leq M + i + 1} | \mu\big)
        \right)
        \\&\leq \lim_{M\to \infty}\sup_n 
        \E_\mu\left(\sum_{i=0}^\infty \E\big((M + i + 1)
        \I_{|\bar{X}| > M + i} | \mu\big) \right)
        \\&= \lim_{M\to \infty}\sup_n 
        \E_\mu\left(\sum_{i=0}^\infty (M + i + 1)
        \Pr\big(|\bar{X}| > M + i | \mu\big) \right)
        \\&= \lim_{M\to \infty}\sup_n 
        \E_\mu\left(\sum_{i=0}^\infty (M + i + 1)
        \bigg(\Pr\big(\bar{X} > M + i | \mu\big) 
        + \Pr\big(\bar{X} < -M - i | \mu\big)\bigg)
        \right)
        \\&\leq \lim_{M\to \infty}\sup_n 
        \E_\mu\left(\sum_{i=0}^\infty (M + i + 1)
        \bigg(e^{-\frac{n(M + i - \mu)^2}{2\sigma_k^2}}
        + e^{-\frac{n(\mu + M + i)^2}{2\sigma_k^2}}\bigg)
        \right)
        \\&\leq \lim_{M\to \infty} \E_\mu\left(\sum_{i=0}^\infty (M + i + 1)
        \bigg(e^{-\frac{(M + i - \mu)^2}{2\sigma_k^2}}
        + e^{-\frac{(\mu + M + i)^2}{2\sigma_k^2}}\bigg)
        \right),
    \end{align*}
    so
    \begin{equation}
        \begin{split}
            \lim_{M\to \infty}\sup_n 
            \E\big(|\bar{X}|\I_{|\bar{X}| > M}\big)
            &\leq
            \lim_{M\to \infty} \E_\mu\left(\sum_{i=0}^\infty (M + i + 1)
            e^{-\frac{(M + i - \mu)^2}{2\sigma_k^2}}\right)
            \\&+ \lim_{M\to \infty} \E_\mu\left(\sum_{i=0}^\infty (M + i + 1)
            e^{-\frac{(\mu + M + i)^2}{2\sigma_k^2}}
            \right).
        \end{split}
        \label{eq:gaussian-convergence-rate-proof-2}
    \end{equation}
    Looking at the first term on the RHS of \eqref{eq:gaussian-convergence-rate-proof-2},
    when $|M + i - \mu| \geq 1$, $(M + i - \mu)^2 \geq M + i - \mu$.
    It is possible that $|M + i - \mu| < 1$ for exactly two 
    values of $i$ that depend on $\mu$. Let $i_{\mu1}$ and $i_{\mu2}$ 
    be those values. 
    We know that $i_{\mu j} < 1 + \mu - M$ for $j\in \{1, 2\}$.
    Now
    \begin{equation}
        \begin{split}
            \E_\mu\left(\sum_{i=0}^\infty (M + i + 1)
            e^{-\frac{(M + i - \mu)^2}{2\sigma_k^2}}
            \right)
            &= \E_\mu\left(\sum_{i=0, i\neq i_{\mu 1}, i\neq i_{\mu 2}}^\infty (M + i + 1)
            e^{-\frac{(M + i - \mu)^2}{2\sigma_k^2}}
            \right)
            \\&+ \E_\mu\left( (M + i_{\mu 1} + 1)
            e^{-\frac{(M + i_{\mu 1} - \mu)^2}{2\sigma_k^2}}
            \right)
            \\&+ \E_\mu\left( (M + i_{\mu 2} + 1)
            e^{-\frac{(M + i_{\mu 2} - \mu)^2}{2\sigma_k^2}}
            \right).
        \end{split}
        \label{eq:gaussian-convergence-rate-proof-1}
    \end{equation}
    We can upper bound the series using $(M + i - \mu)^2 \geq M + i - \mu$
    and the formulas for geometric series.
    \begin{align*}
        &\E_\mu\left(\sum_{i=0, i\neq i_{\mu 1}, i\neq i_{\mu 2}}^\infty (M + i + 1)
        e^{-\frac{(M + i - \mu)^2}{2\sigma_k^2}}
        \right)
        \\&\leq \E_\mu\left(\sum_{i=0, i\neq i_{\mu 1}, i\neq i_{\mu 2}}^\infty (M + i + 1)
        e^{-\frac{(M + i - \mu)}{2\sigma_k^2}}
        \right)
        \\&\leq \E_\mu\left(\sum_{i=0}^\infty (M + i + 1)
        e^{-\frac{(M + i - \mu)}{2\sigma_k^2}}
        \right)
        \\&= \E_\mu\left(\sum_{i=0}^\infty (M + i + 1)
        e^{-\frac{(M - \mu)}{2\sigma_k^2}}
        \left(e^{-\frac{1}{2\sigma_k^2}}\right)^i
        \right)
        \\&= \E_\mu\left(
        \sum_{i=0}^\infty M
        e^{-\frac{(M - \mu)}{2\sigma_k^2}}
        \left(e^{-\frac{1}{2\sigma_k^2}}\right)^i
        \right)
        + \E_\mu\left(
        \sum_{i=0}^\infty (i + 1)
        e^{-\frac{(M - \mu)}{2\sigma_k^2}}
        \left(e^{-\frac{1}{2\sigma_k^2}}\right)^i
        \right)
        \\&= \E_\mu\left(
        \frac{Me^{-\frac{(M - \mu)}{2\sigma_k^2}}}{1 - e^{-\frac{1}{2\sigma_k^2}}}
        \right)
        + \E_\mu\left(
        \frac{Me^{-\frac{(M - \mu)}{2\sigma_k^2}}}{\left(1 - e^{-\frac{1}{2\sigma_k^2}}\right)^2}
        \right)
        \\&= \left(\frac{M}{1 - e^{-\frac{1}{2\sigma_k^2}}}
        + \frac{M}{\left(1 - e^{-\frac{1}{2\sigma_k^2}}\right)^2}\right)
        \E_\mu\left(
        e^{-\frac{(M - \mu)}{2\sigma_k^2}}
        \right).
    \end{align*}
    For the expectation, we have
    \begin{align*}
        \E_\mu\left(e^{-\frac{M - \mu}{2\sigma_k^2}} \right)
        = e^{-\frac{M}{2\sigma_k^2}} 
        \E_\mu\left(e^{\frac{1}{2\sigma_k^2}\mu} \right).
    \end{align*}
    $\E_\mu\left(e^{\frac{1}{2\sigma_k^2}\mu} \right)$ is finite, as it is 
    an evaluation of the moment generating function of $\mu$, which means that
    \begin{equation*}
        \lim_{M\to \infty} \left(\frac{M}{1 - e^{-\frac{1}{2\sigma_k^2}}}
        + \frac{M}{\left(1 - e^{-\frac{1}{2\sigma_k^2}}\right)^2}\right)
        \E_\mu\left(
        e^{-\frac{(M - \mu)}{2\sigma_k^2}}
        \right)
        = 0.
    \end{equation*}
    For the two other terms on the RHS of \eqref{eq:gaussian-convergence-rate-proof-1}
    \begin{align*}
        \E_\mu\left( (M + i_{\mu j} + 1)
        e^{-\frac{(M + i_{\mu j} - \mu)^2}{2\sigma_k^2}}
        \right)
        &\leq \E_\mu\left( (M + 1 + \mu - M + 1)
        e^{-\frac{(M + i_{\mu j} - \mu)^2}{2\sigma_k^2}}
        \right)
        \\&= \E_\mu\left( (\mu + 2)
        e^{-\frac{(M + i_{\mu j} - \mu)^2}{2\sigma_k^2}}
        \right)
        \\&\to 0
    \end{align*}
    as $M\to \infty$ by the dominated convergence theorem, as 
    \begin{equation*}
        (\mu + 2)e^{-\frac{(M + i_{\mu j} - \mu)^2}{2\sigma_k^2}}
        \leq (\mu + 2)e^{-\frac{0}{2\sigma_k^2}},
    \end{equation*}
    and the right-hand-side has a finite expectation.

    We have now shown that the first limit on the RHS of 
    \eqref{eq:gaussian-convergence-rate-proof-2} is 0. For the other limit, 
    setting $\mu' = -\mu$, and using the 
    reasoning above with $\mu$ replaced by $\mu'$, we have 
    \begin{align*}
        &\lim_{M\to \infty} \E_\mu\left(\sum_{i=0}^\infty (M + i + 1)
        \bigg(e^{-\frac{(\mu + M + i)^2}{2\sigma_k^2}}\bigg)
        \right)
        \\&= \lim_{M\to \infty} \E_{\mu'}\left(\sum_{i=0}^\infty (M + i + 1)
        \bigg(e^{-\frac{(M + i - \mu')^2}{2\sigma_k^2}}\bigg)
        \right)
        \\&= 0.
    \end{align*}
    so the RHS of \eqref{eq:gaussian-convergence-rate-proof-2} is 0.
    
    We have now shown that 
    \begin{equation*}
        \lim_{M\to \infty}\sup_n 
        \E\big(|\bar{X}|\I_{|\bar{X}| > M}\big)
        = 0,
    \end{equation*}
    or, in other words, that $|\bar{X}|$ is uniformly integrable. As shown earlier, 
    this concludes the proof that 
    \begin{equation*}
        \sqrt{n}\TV\big(\qap_n, D_n\big)
    \end{equation*}
    is uniformly integrable when 
    $\xpred_n \sim p(\xpred_n | \xreal)$.

    To show that $\sqrt{n}\TV\big(\qex_n, D_n\big)$ is uniformly integrable, 
    as in the proof of Lemma~\ref{thm:bvm-implies-condition}, we have
    \begin{equation*}
        p(Q | \xreal, \xpred) \propto p(\xpred | Q)p(\xreal | Q)p(Q),
    \end{equation*}
    so we can view both $p(Q | \xreal, \xpred_n)$ and $p(Q | \xpred_n)$ as the posteriors for the 
    same Bayesian inference problem with observed data $\xpred$, and priors
    $p(Q | \xreal) \propto p(\xreal | Q)p(Q)$ and $p(Q)$, respectively. 
    $p(Q | \xreal)$ is Gaussian, so the uniform integrability of 
    \begin{equation*}
        \sqrt{n}\TV\big(\qex_n, D_n\big)
    \end{equation*}
    follows from the previous case with different values for $\mu_0$ and $\sigma_0^2$.
\end{proof}

\section{Sampling the Exact Posterior in the Toy Data Experiment}\label{sec:toy-data-sampling-exact-posterior}
In order to sample the exact 
posterior $p(Q | \sdp)$, we use another decomposition:
\begin{equation}
    p(Q | \sdp) 
    = \int p(Q | \sdp, \xreal)p(\xreal | \sdp)\dx \xreal
    = \int p(Q | \xreal)p(\xreal | \sdp)\dx \xreal,
\end{equation}
where $p(Q | \sdp, \xreal) = p(Q | \xreal)$ due to the independencies of the graphical 
model in Figure~\ref{fig:random_variables}. It remains to sample 
$p(X | \sdp)$. This is not tractable in general, but is possible in 
the toy data setting due to using the full set of 3-way marginals that 
covers all possible values of a datapoint, and the simplicity of the toy 
data. 

We can decompose 
\begin{equation*}
    p(X | \sdp) 
    = \int p(s | \sdp)p(X | s) \dx \theta \dx X
    = \int p(X | s)\int p(s, \theta | \sdp) \dx \theta \dx X,
\end{equation*}
so we can sample $(s, \theta) \sim p(s, \theta | s)$ and then sample 
$\xreal \sim p(\xreal | s)$ to obtain a sample from $p(\xreal | \sdp)$.
Due to the simplicity of the toy data setting, sampling both 
$p(s, \theta | s)$ and $p(\xreal | s)$ is possible.

NAPSU-MQ uses the following Bayesian inference problem:
\begin{align*}
    \theta &\sim \mathrm{Prior} \\
    X &\sim \medt^n \\
    s &= a(X) \\
    \sdp &\sim \caln(s, \sigmadp^2),
\end{align*}
where $a$ are the marginal queries, $\sigmadp^2$ is the noise variance 
of the Gaussian mechanism, and $\medt^n$ is the maximum entropy 
distribution~\citep{raisaNoiseAwareStatisticalInference2023}
with point probability 
\begin{equation*}
    p(x) = \exp(\theta^T a(x) - \theta_0(\theta)),
\end{equation*}
where $\theta_0$ is the log-normalising constant.

In the toy data setting, $a$ is the full set of 3-way marginals for all of the 
3 variables. In other words, $a(x)$ is the one-hot coding of $x$, so 
$s = a(X)$ is a vector of counts of how many times each of the 8 possible
values is repeated in $X$. This means that sampling $p(X | s)$ is simple:
\begin{enumerate}
    \item For each possible value of a datapoint, find the corresponding 
    count from $s$, and repeat that datapoint according the that count.
    \item Shuffle the datapoints to a random order.
\end{enumerate}
As the downstream analysis $p(Q | \xreal)$ doesn't depend on the order of 
the datapoints, the second step is not actually needed.

To sample $p(s, \theta | \sdp)$, we use a Metropolis-within-Gibbs 
sampler~\citep{gilksAdaptiveRejectionMetropolis1995}
that sequentially updates $s$ and $\theta$ while keeping the other fixed.
The proposal for $\theta$ is obtained from Hamiltonian Monte Carlo 
(HMC)~\citep{duaneHybridMonteCarlo1987,nealMCMCUsingHamiltonian2011}.
The proposal for $s$ is obtained by repeatedly choosing a random index in 
$s$ to increment and another to decrement. It is possible to obtain negative 
values in $s$ from this proposal, but those will always be rejected by the 
acceptance test, as the likelihood for them is 0.

To initialise the sampler, we pick an initial value for $\theta$ from a 
Gaussian distribution, and pick the initial $s$ by rounding $\sdp$ 
to integer values, changing the rounded values such that they sum to $n$
while ensuring that all values are non-negative.

The step size for the HMC we used is 0.05, and the number of steps is 20.
In the $s$ proposal, we repeat the combination of an increment and a decrement 
30 times. We take 20000 samples in total from 4 parallel chains, and drop the 
first 20\% as warmup samples.

The method described in this section is similar to the noise-aware Bayesian 
inference method of \citet{juDataAugmentationMCMC2022}. The difference 
between the two is that \citet{juDataAugmentationMCMC2022} use $X$ instead of 
$s$ as the auxiliary variable, and they sample the $X$ proposals from the model,
changing one datapoint at a time. This makes their algorithm more generalisable.

\bibliography{bayes-downstream-analysis-references}

\end{document}